\documentclass[a4paper,UKenglish]{lipicsarxiv}
\usepackage{microtype}%if unwanted, comment out or use option "draft"

\title{Register transducers are marble transducers}

\author{Gaëtan Douéneau-Tabot}{IRIF, Université de Paris}{doueneau@irif.fr}{}{}%mandatory, please use full name; only 1 author per \author macro; first two parameters are mandatory, other parameters can be empty.

\author{Emmanuel Filiot}{Université libre de Bruxelles \& F.R.S.-FNRS}{efiliot@ulb.ac.be}{}{}

\author{Paul Gastin}{LSV, ENS Paris-Saclay, CNRS, Universit{\'e}
Paris-Saclay}{paul.gastin@ens-paris-saclay.fr}{https://orcid.org/0000-0002-1313-7722}{}

\authorrunning{G. Douéneau-Tabot and E. Filiot and P. Gastin}%mandatory. First: Use abbreviated first/middle names. Second (only in severe cases): Use first author plus 'et al.'

\Copyright{Gaëtan Douéneau-Tabot, Emmanuel Filiot, Paul Gastin}%mandatory, please use full first names. LIPIcs license is "CC-BY";  http://creativecommons.org/licenses/by/3.0/

\subjclass{Theory of computation $\rightarrow$ Transducers}% mandatory: Please choose ACM 2012 classifications from https://www.acm.org/publications/class-2012 or https://dl.acm.org/ccs/ccs_flat.cfm . E.g., cite as "General and reference $\rightarrow$ General literature" or \ccsdesc[100]{General and reference~General literature}. 

%mandatory

\category{}%optional, e.g. invited paper

\relatedversion{arXiv}{}%optional, e.g. full version hosted on arXiv, HAL, or other respository/website

\supplement{}%optional, e.g. related research data, source code, ... hosted on a repository like zenodo, figshare, GitHub, ...

\funding{}%optional, to capture a funding statement, which applies to all authors. Please enter author specific funding statements as fifth argument of the \author macro.

\acknowledgements{}%optional

\EventEditors{}
\EventNoEds{0} 
\EventLongTitle{}
\EventShortTitle{}
\EventAcronym{}
\EventYear{}
\EventDate{}
\EventLocation{}
\EventLogo{}
\SeriesVolume{}
\ArticleNo{}
\usepackage{tikz}
\usepackage[utf8]{inputenc}
\usetikzlibrary{shapes,arrows,positioning,automata,shadows}
\usepackage{amsfonts, amsmath, amssymb,dsfont, mathrsfs}
\usepackage{amsthm}
\usepackage{hyperref}
\usepackage[ruled]{algorithm2e}
\usepackage{stmaryrd}

\newtheorem{claim}[theorem]{Claim}
\newtheorem{proposition}[theorem]{\bfseries Proposition}

\newcommand{\dom}{\operatorname{\normalfont{dom}}}

\newcommand{\vide}{\varepsilon}

\newcommand{\lmove}{\triangleleft}
\newcommand{\rmove}{\triangleright}
\newcommand{\lmark}{\vdash}
\newcommand{\rmark}{\dashv}

\newcommand{\ske}{\operatorname{\normalfont \textsf{ske}}}
\newcommand{\beg}{\operatorname{\normalfont \textsf{beg}}}
\newcommand{\fol}{\operatorname{\normalfont \textsf{fol}}}

\newcommand{\regs}{\mf{X}}
\newcommand{\regu}{\mf{U}}
\newcommand{\regp}{\mf{R}}

\newcommand{\oras}{\mf{F}}
\newcommand{\exte}{\mf{f}}
\newcommand{\exteg}{\mf{g}}

\newcommand{\subst}[2]{\mf{S}_{#1}^{#2}}

\newcommand{\trans}{\mc{T}}
\newcommand{\srans}{\mc{S}}
\newcommand{\utrans}{\mc{U}}

\newcommand{\fauto}[1]{\widetilde{#1}}
\newcommand{\auto}{\mc{A}}
\newcommand{\inia}{\alpha}
\newcommand{\fina}{\beta}
\newcommand{\weia}{\mu}
\newcommand{\staa}{\mc{Q}}

\newcommand{\val}{\operatorname{ \normalfont Value}}

\newcommand{\PTIME}{\operatorname{\normalfont{\textsf{PTIME}}}}
\newcommand{\EXP}{\operatorname{\normalfont{\textsf{EXPTIME}}}}

\newcommand{\graph}{\mf{G}}

\newcommand{\mul}{\operatorname{\textsf{mul}}}
\newcommand{\mar}{\operatorname{\normalfont \textsf{marked}}}

\newcommand{\first}{\operatorname{\normalfont first}}
\newcommand{\nnext}{\operatorname{\normalfont next}}
\newcommand{\pow}{\operatorname{\normalfont \textsf{exp}}}
\newcommand{\powo}{\operatorname{\normalfont \textsf{pow}}}
\newcommand{\drop}{\operatorname{\normalfont{\textsf{drop}}}}
\newcommand{\lift}{\operatorname{\normalfont{\textsf{lift}}}}
\newcommand{\sq}{\operatorname{\textsf{square}}}
\newcommand{\mirror}{\operatorname{\textsf{reverse}}}

\newcommand{\SSTF}{{\textsf{SST-F}}}
\newcommand{\NSSTF}{{\textsf{NSST-F}}}
\newcommand{\SST}{{\textsf{SST}}}
\newcommand{\SSTs}{{\textsf{SSTs}}}

\newcommand{\mb}[1]{\mathbb{#1}}
\newcommand{\mc}[1]{\mathcal{#1}}
\newcommand{\mf}[1]{\mathfrak{#1}}

\begin{document}

\maketitle

\begin{abstract} 
  Deterministic two-way transducers define the class of regular functions from
  words to words.  Alur and Cern\'y introduced an equivalent model of 
  transducers with registers called
%   is described in \cite{alur2010expressiveness} under the name of 
  copyless streaming string transducers.  In this paper, we drop the "copyless"
  restriction on these machines and show that they are equivalent to two-way
  transducers enhanced with the ability to drop marks, named "marbles",
%   \cite{engelfriet1999trips}, 
  on the input.  We relate the maximal number of marbles used with the amount of
  register copies performed by the streaming string transducer.  Finally, we
  show that the class membership problems associated with these models are
  decidable.  Our results can be interpreted in terms of program optimization
  for simple recursive and iterative programs.
\end{abstract}

%%%%%%%%%%%%%%%%%%%%%%%%%%
% Introduction
%%%%%%%%%%%%%%%%%%%%%%%%%%

%%
%% Main
%%

\section{Introduction}

Regular languages have been a cornerstone of theoretical computer science since the 1950's. They can be described by several equivalent models such as deterministic, non-deterministic, or two-way (the reading head can move in two directions) finite automata \cite{shepherdson1959reduction}.

A natural extension consists in adding an output mechanism to finite automata.
Such machines, called \emph{transducers}, describe functions from words to words
(or relations when non-deterministic).  In this case, the landscape generally
becomes more complex, as noted in 1967 by D. Scott: \flqq{} the functions
computed by the various machines are more important - or at least more basic -
than the sets accepted by these device \frqq{} \cite{scott1967some}.
Furthermore, transducers provide a natural way to model simple programs that
produce outputs.

\subparagraph*{Regular functions and copyless register transducers.} The
particular model of \emph{two-way transducer} consists in a two-way automaton
enhanced with an output function.  It describes the class of \emph{regular
functions} which has been intensively studied for its natural properties:
closure under composition \cite{chytil1977serial}, logical characterization by
monadic second-order transductions \cite{engelfriet2001mso}, decidable
equivalence problem \cite{gurari1982equivalence}, etc.

In \cite{alur2010expressiveness}, the equivalent model of \emph{copyless
streaming string transducer} ($\SST$) is described.  This machine processes its
input in a one-way fashion, while storing pieces of their output in a finite set
of registers: it is at the same time simpler (since it reads the input only
once) and more complex (since it needs registers) than a two-way transducer.
Registers are updated by simple concatenation operations.  However, the content
of a register can never be duplicated ("copyless"), which allows one to
implement the model efficiently for a streaming use.

\subparagraph*{Copyful register transducers.} Regular functions remain quite
limited in terms of expressiveness, since the size of the output can be at most
linear in the input's.  In this paper, we study the class of functions computed
by copyful $\SSTs$, i.e. register transducers that can duplicate their registers.
With this model, it becomes possible to produce outputs that have a polynomial,
or even exponential, size in the input's.  Meanwhile, it preserves many "good
properties" of regular functions, such as decidability of the equivalence
problem \cite{filiot2017copyful}.

\subparagraph*{Marble transducers.} Our first objective is to extend the
aforementioned correspondence between copyless $\SSTs$ and two-way transducers,
by providing a model of transducer without registers that is equivalent to
copyful $\SSTs$.  For this, we define \emph{marble transducers} (introduced for
trees in \cite{engelfriet1999trips}) and show equivalence.  This model consists
in a two-way transducer that can drop several marks ("marbles") on its input,
following a stack discipline.  Indeed, new marbles can only be dropped on the
left of the positions already marked.  Informally, our result shows that copyful
$\SSTs$ correspond to some recursive algorithms (hence the stack).

A very natural way to restrict the power of marble transducers is to bound the number of marks that can be used. We define \emph{$k$-marble transducer} that can use at most $k$ marks. Intuitively, it corresponds to iterative algorithms with "for" loops, such that the maximal depth of nested loops is $k+1$. In particular, $0$-marble transducers are exactly two-way transducers (since they use no marbles). Whereas marble transducers in general can have an exponential execution time, a $k$-marble transducer runs in polynomial time, more precisely $\mc{O}(n^{k+1})$ when $n$ is the input's length. Hence, it produces outputs of size $\mc{O}(n^{k+1})$.

As a second main result, we show that \emph{$k$-marble transducers} are
equivalent to a model of \emph{$k$-layered \SSTs}, i.e. $\SSTs$ with hierarchical
restrictions on their copies.  In particular for $k=0$, we recover the
correspondence between two-way transducers and copyless $\SSTs$.

\subparagraph*{Optimization and class membership problems.} As evoked above, our models of marble transducers have at most an exponential complexity (or "execution time"), but it becomes polynomial if we restrict the number of marks used. In practice, a natural question is that of optimization: can we transform an exponential algorithm in a polynomial equivalent one? Can we reach the smallest possible complexity? Having a tool to optimize programs is of foremost interest since it allows to write naive algorithms without worrying about the complexity. Due to well-known undecidability statements, optimizing any algorithm is hopeless in theory, thus having results for a "regular" kernel is already interesting. 

From a theoretical point of view, the optimization problem is known as
(effective) \emph{class membership problem}.  It instantiates as follows: given
a function computed by a marble transducer, can it be computed by a $k$-marble
transducer?  An easy lower bound is given by the size of the output, since for
instance we cannot produce a string of size $\Omega(n^2)$ with a two-way
transducer.  As shown in our third main result, it is in fact a sufficient
criterion to decide membership: a function from our class is computable with $k$
marbles if and only if it grows in $\mc{O}(n^{k+1})$ (and this property is
decidable).  This result shows the robustness of our $k$-marble model, since
a simple syntactical restriction is sufficient to describe a semantical
property.  Its proof is the most involved of this paper; it uses the
correspondence with $\SSTs$.

Similar optimization results have recently been obtained in
\cite{lhote2020pebble} for the class of \emph{polyregular functions}, defined
using $k$-pebble transducers \cite{bojanczyk2018polyregular} (an extension of
$k$-marble).  Interestingly, their conclusion is very similar to ours, that is:
an output of size $\mc{O}(n^k)$ can always be produced using $k$ nested loops.
We shall discuss in conclusion how our results both refine and extend theirs.
Contrary to us, the equivalence problem is an open problem for their model.
Furthermore, pebble transducers have never been related to a class of streaming
algorithms, contrary to what we show for marbles.

\subparagraph*{Outline.} After recalling in Section \ref{sec:prelim} the basic
definitions of two-way transducers and $\SSTs$, we present in Section
\ref{sec:marbles} the model of marble transducer and show equivalence.  We then
study in Section \ref{sec:bounded} the case of $k$-marble transducers and relate
them to specific $\SSTs$.  Finally, we solve in Section \ref{sec:membership} the
class membership problems associated with these models.  Due to space
constraints, several proofs are only sketched; we chose to focus on the proofs
of the last section since they describe an algorithm for program optimization.

\section{Preliminaries}

\label{sec:prelim}

We denote by $\mb{N}$ the set of nonnegative integers. Capital letters $A$ and $B$ are used to denote alphabets, i.e. finite sets of letters. If $w \in A^*$ is a word, let $|w| \in\mb{N}$ be its length, and for $1 \le m \le |w|$ let $w[m]$ be its $m$-th letter. The empty word is denoted $\vide$. If $1 \le m \le n \le |w|$, let $w[m{:}n] = w[m]w[m+1] \cdots w[n]$. We assume that the reader is familiar with the basics of automata theory, and in particular the notions of one-way and two-way deterministic automata (see e.g. \cite{shepherdson1959reduction}).

\subparagraph*{Two-way transducers.} A deterministic two-way transducer is a deterministic two-way automaton enhanced with the ability to produce outputs along its run.  The class of functions described by these machines is known as "regular functions" \cite{chytil1977serial, engelfriet2001mso}.

\begin{definition} A \emph{(deterministic) two-way transducer} $\trans = (A,B,Q,q_0, \delta, \lambda,F)$ consists of:
\begin{itemize}
\item an input alphabet $A$, an output alphabet $B$;
\item a finite set of states $Q$, with an initial state $q_0 \in Q$  and a set of final states $F \subseteq Q$ ;
\item a (partial) transition function $\delta: Q \times (A\uplus\{  \lmark,\rmark) \rightarrow Q \times \{\lmove, \rmove\}$;
\item a (partial) output function $\lambda: Q \times (A\uplus\{ \lmark,\rmark\} ) \rightarrow B^*$ with same domain as $\delta$.
\end{itemize}
\end{definition}

When given as input a word $w \in A^*$, the two-way transducer disposes of a read-only input tape containing $\lmark w \rmark$. The marks $\lmark$ and $\rmark$ are used to detect the borders of the tape, by convention we denote them as positions $0$ and $|w|+1$ of $w$.

Formally, a \emph{configuration} over  $\lmark w \rmark$ is a tuple $(q,m)$ where $q \in Q$ is the current state and $0 \le m \le |w|+1$ is the position of the reading head. The \emph{transition relation} $\rightarrow$ is defined as follows. Given a configuration $(q,m)$, let $(q',\star):= \delta(q,w[m])$. Then $(q, m) \rightarrow (q', m')$ whenever either $\star = \lmove$ and  $m' = m-1$ (move left), or $\star = \rmove$ and $m' = m+1$ (move right), with $0 \le m' \le |w|+1$. A \emph{run} is a sequence of configurations following $\rightarrow$. Accepting runs are those that begin in $(q_0, 0)$ and end in a configuration of the form $(q, |w|+1)$ with $q \in F$.

The (partial) function $f : A^* \rightarrow B^*$ computed by the machine is defined as follows. If there exists an accepting run on $\lmark w \rmark$, then it is unique and $f(w)$ is the concatenation of all  the $\lambda(q,w[m])$ along the transitions of this run. Otherwise $f(w)$ is undefined.

\begin{example} \label{ex:mirror} Let $\mirror : A^* \rightarrow A^*$ be the function that maps a word $abac$ to its mirror image $caba$. It can be performed by a two-way transducer that first goes to the right symbol $\rmark$, and then reads $w$ from right to left while outputting the letters.
\end{example}

\subparagraph*{Streaming string transducers.} Informally, a \emph{streaming
string transducer} \cite{alur2010expressiveness} is a one-way deterministic
automaton with a finite set $\regs$ of registers that store strings over the
output alphabet $B$.  These registers are modified using \emph{substitutions},
i.e. mappings $\regs \rightarrow (B \uplus \regs)^*$.  We denote by
$\subst{\regs}{B}$ the set of these substitutions.  They can be extended
morphically from $(B \uplus \regs)^*$ to $ (B \uplus \regs)^*$ by preserving the
elements of $B$.  As explained in Example \ref{ex:subst}, they can be composed
by setting $(s_1 \circ s_2)(x) := s_1 (s_2(x))$ for $x \in \regs$.

\begin{example} \label{ex:subst} Let $\regs = \{x,y\}$ and $B = \{b\}$. Consider the substitutions $s_1:= x \mapsto b, y \mapsto bxyb$ and $s_2:= x \mapsto xb, y \mapsto xy$, then $ s_1 \circ s_2 (x) = s_1(xb) = bb$ and $ s_1 \circ s_2 (y) = s_1(xy) = bbxyb$.
\end{example}

\begin{definition}A \emph{streaming string transducer} ($\SST$) $\trans = (A,B, Q,\regs,q_0, \iota, \delta, \lambda,F)$ is:
\begin{itemize}

\item an input alphabet $A$ and an output alphabet $B$;

\item a finite set of states $Q$ with an initial state $q_0 \in Q$;

\item a finite set $\regs$ of registers;

\item an initial function $\iota : \regs \rightarrow B^*$;

\item a (partial) transition function $\delta : Q \times A \rightarrow Q$;

\item a (partial) register update function $\lambda: Q \times A \rightarrow \subst{\regs}{B}$  with same domain as $\delta$;

\item a (partial) output function  $F: Q \rightarrow (\regs \cup B)^*$.

\end{itemize}
\end{definition}

This machine defines a (partial) function $f : A^* \rightarrow B^*$ as follows. Let us fix $w \in A^*$. If there is no accepting run of the \emph{one-way} automaton $(A, Q,q_0, \delta, \dom(F))$ over $w$, then $f(w)$ is undefined. Otherwise, let $q_m:= \delta(q_0,w[1{:}m])$ be the $m$-th state of this run. We define for $0 \le m \le |w|$, $\trans^{w[1{:}m]} \in \subst{\regs}{B}$ ("the values of the registers after reading $w[1{:}m]$") as follows:
\begin{itemize}

\item $\trans^{w[1:0]} (x) = \iota(x)$ for all $x \in \regs$; \label{mark:valu}

\item for $1 \le m \le |w|$, $\trans^{w[1{:}m]} := \trans^{w[1{:}(m-1)]}\circ\lambda(q_m, w[m])$. This formula e.g. means that if $\trans^{w[1{:}(m-1)]}(x) = ab$ and $\lambda(q_m, w[m])(x) = xx$, then $\trans^{w[1{:}m]}(x) = abab$.

\end{itemize}
In this case, we set $f(w) := \trans^{w}(F(q_{|w|})) \in B^*$. In other words, we combine the final values of the registers following the output function.

\begin{example} \label{ex:mirror:sst} 
  The $\mirror$ of Example \ref{ex:mirror} can be computed by an $\SST$ with
  one state and one register $x$.  When seeing a letter $a$, the $\SST$ updates
  $x\mapsto ax$ ($a$ is added in front of $x$).
\end{example}

\begin{example} \label{ex:pow} Consider the function $\pow:a^n \mapsto a^{2^n}$. It is computed by an $\SST$ with one register $x$ initialized to $a$ and updated $x \mapsto xx$ at each transition.
\end{example}

The function $\pow$ of Example \ref{ex:pow} cannot be computed by a
deterministic two-way transducer.  Indeed, a two-way transducer computing a
function $f$ has only $|Q|(|w|+2)$ possible configurations on input $w$,
therefore we must have $|f(w)| = \mc{O}(|w|)$.

In order to make two-way transducers and $\SSTs$ coincide, the solution of
\cite{alur2010expressiveness} is to forbid duplications of registers.  A
substitution $\sigma \in \subst{\regs}{B}$ is said to be \emph{copyless} if each
register $x \in \regs$ appears at most once in the whole set of words
$\{\sigma(x)\mid x \in \regs\}$.  The substitution $s_1$ of Example
\ref{ex:subst} is copyless whereas $s_2$ is not.  An $\SST$ is said to be copyless
whenever it uses only copyless substitutions for the $\lambda(q,a)$; the $\SST$
of Example \ref{ex:mirror:sst} is so.

\begin{theorem}[\cite{alur2010expressiveness, dartois2017reversible}] \label{theo:alur} 
  Two-way transducers and copyless $\SSTs$ describe the same class of functions
  ("regular functions").  The right to left conversion is effective in $\PTIME$,
  and the converse one in $\EXP$.
\end{theorem}

\begin{remark}
For the complexities, the "size" of the machines is that of a reasonable representation. For a two-way transducer, it is roughly the total size of the outputs that label its transitions. For an $\SST$, it is the total size of its substitutions.
\end{remark}

\section{Marble transducers and streaming string transducers}

\label{sec:marbles}

As evoked in the introduction, our first goal is to extend Theorem
\ref{theo:alur} by describing a machine without registers that captures the
expressiveness of $\SSTs$ with copies.  For this purpose, we shall use a variant
of two-way transducers that can drop/lift several marks on their input.
However, the use of marks has to be strongly restricted so that the machine is
not too expressive (see e.g. \cite{ibarra1971characterizations}).  The model we
propose here, named \emph{marble transducer} after \cite{engelfriet1999trips},
can drop marks ("marbles") of different colors and the reading head has to stay 
on the left of marbles:
% only on the right of its reading head.  In other words, 
a stack of marks is stored on the % rightmost part of the
input and if the machine wants to move forward from a position where there is a
marble, it first has to remove it.

\begin{definition} \label{def:manymarbles}

A \emph{(deterministic) marble transducer} $\trans = (A,B,Q,C,q_0,\delta,\lambda,F)$ consists of:

\item

\begin{itemize}

\item an input alphabet $A$;

\item a finite set of states $Q$ with an initial state $q_0 \in Q$ and a set $F \subseteq Q$ of final states;

\item a finite set of marble colors $C$;

\item a transition function $\delta: Q \times (A \uplus \{\lmark, \rmark\}) \times (C \uplus \{\varnothing\}) \rightarrow Q \times (\{\lmove, \rmove, \lift\} \uplus \{\drop_c\mid c \in C\})$ such that $\forall q\in Q, a \in A, c \in C$ we have $\delta(q,a,c) \in Q \times \{ \lmove,\lift \}$ (we cannot move right nor drop another marble when we see a marble).

\item an output function $\lambda: Q \times (A \uplus \{\lmark, \rmark\}) \times (C \uplus \{\varnothing\})  \rightarrow B^*$ with same domain as $\delta$.

\end{itemize}

\end{definition}

As for two-way transducers, the symbols $\lmark$ and $\rmark$ are used to denote
the borders of the input.  A \emph{configuration} over $\lmark w \rmark$ is a
tuple $(q,m, \pi)$ where $q \in Q$ is the current state, $m$ is the position of
the reading head, and $\pi = (c_{\ell},m_{\ell}) \cdots (c_1,m_1)$ is the stack
of the positions and colors of the $\ell$ marbles dropped (hence $\ell\geq0$ and
$0 \le m \le m_{\ell} < \cdots < m_1 \le |w|+1$ and $c_i \in C$).  An example of
configuration is depicted in Figure \ref{fig:config} below.

\begin{figure}[h!]

\begin{center}
\begin{tikzpicture}[scale=1]
	\node (in) at (-2.3,0) [right]{\small Input word};
	\fill[fill = blue,even odd rule] (2,0.5) circle (0.2);
	\fill[fill = red,even odd rule] (4,0.5) circle (0.2);
	\fill[fill = blue,even odd rule] (5,0.5) circle (0.2);
	\node[above] (in) at (0,-0.25) []{  $\lmark$};
	\node[above] (in) at (1,-0.25) []{  $b$};
	\node[above] (in) at (2,-0.25) []{  $a$};
	\node[above] (in) at (3,-0.25) []{  $b$};
	\node[above] (in) at (4,-0.25) []{  $b$};
	\node[above] (in) at (5,-0.25) []{  $b$};
	\node[above] (in) at (6,-0.25) []{  $a$};
	\node[above] (in) at (7,-0.25) []{  $\rmark$};
	\draw[](1.5,-0.5) rectangle (2.5,1);
        \draw[-](2,-0.5) to[out= -90, in = 90] (4,-1.25);
        	\node (in) at (2,1.25) []{\small Reading head};
	\node (in) at (3.9,-1.4) []{\small Control state $q$};
       \end{tikzpicture}
\end{center}

\vspace*{-0.3cm}
\caption{\label{fig:config} Configuration $(q, 2, (\textcolor{blue}{\bullet}, 2) (\textcolor{red}{\bullet}, 5) (\textcolor{blue}{\bullet}, 6))$ of a marble transducer over $babbba$. Note that allowed transitions starting from this configuration are either $\lift$ or $\lmove$.}

\end{figure}

The \label{mark:transrel} \emph{transition relation} $\rightarrow$ of $\trans$ is defined as follows. Given a configuration $(q,m, \pi)$, let $k := c$ if $\pi = (c,m) \cdots$ (marble $c$ in position $m$) and $k := \varnothing$ otherwise (no marble in $m$). Let $(q',\star):= \delta(q,w[m], k)$. Then $(q, m, \pi) \rightarrow (q', m', \pi')$ whenever one of the following holds:
\begin{itemize}

\item move left: $\star = \lmove$,  $m' = m-1 \ge 0$ and $\pi = \pi'$;

\item move right: $\star = \rmove$, $m' = m+1 \le |w|+1$ and $\pi = \pi'$ (only when $k = \varnothing$);

\item lift a pebble:  $\star = \lift$, $m = m'$ and $\pi = (c,m) \pi' $  (only when $k \neq \varnothing$);

\item drop a pebble: $\star = \drop_c$, $m = m'$, $\pi' = (c,m) \pi $  (only when $k = \varnothing$).

\end{itemize}
The notion of \emph{run} is defined as usual with $\rightarrow$. Accepting runs are finite runs that begin in $(q_0, 0, \varepsilon)$ and end in a configuration of the form $(q, |w|+1, \varepsilon)$ with $q \in F$.

The partial function $f : A^* \rightarrow B^*$ computed by the machine is defined as follows. If there exists an accepting run on $\lmark w \rmark$, then it is unique and $f(w)$ is the concatenation of all  the $\lambda(q,w[m], k)$ along the transitions of this run. Otherwise $f(w)$ is undefined.

\begin{remark} In case no marbles are used, the machine is simply a two-way transducer.
\end{remark}

As they can have an exponential number of configurations, marble transducers can
produce outputs of exponential size, like $\SSTs$ (see Example \ref{ex:expo:mt}).

\begin{example} \label{ex:expo:mt} 
  The function $\pow : a^n \mapsto a^{2^n}$ of Example \ref{ex:pow} is computed
  by a marble transducer with $C = \{0,1 \}$.  The idea is to use marbles to
  count in binary on the input $a^n$.  We first write $0^n$ on the input, then
  we increment it to $1{0}^{n-1}$, then $010^{n-2}$, $110^{n-2}$, $\dots$, $1^n$
  (there are $2^n$ numbers).  These increments can be done while preserving the
  stack discipline of the marbles: we move right and lift while we see
  $1$'s, when a $0$ is met we replace it by $1$, then we move left dropping
  $0$'s.  Initially and after each increment, we output an $a$ to produce $a^{2^n}$.
\end{example}

We are now ready to state our first generalization of Theorem \ref{theo:alur}. For the complexity, the size of a marble transducer is that of its output labels plus its number of marbles.

\begin{theorem} \label{theo:sstmt}
  Marble transducers and $\SSTs$ describe the same class of functions.  The
  right to left conversion is effective in $\PTIME$, and the converse one in
  $\EXP$.
\end{theorem}

\begin{proof}[Proof sketch.]  Let $\trans = (A,B,Q, C,q_0,\delta, \lambda,F)$ be
a marble transducer.  We simulate it with an $\SST$ by adapting the classical
reduction from two-way automata to one-way automata via \emph{crossing
sequences} \cite{shepherdson1959reduction}.  When in position $m$ of input $w
\in A^*$, the $\SST$ keeps track of the right-to-right runs of the marble
transducer on the prefix $w[1{:}m]$.  This abstraction is updated at each new
letter by considering the transitions it induces.  Due to the presence of
marbles, the same right-to-right run can be executed multiple times, but with
different stack of marbles (thus avoiding looping behaviors).  These multiple
similar executions are handled using copies in the $\SST$ model.  From $\SSTs$ to
marble transducers, we execute a recursive algorithm to compute the contents of
the registers, and implement it with a marble transducer using the marbles to
code the stack of calls (recursive calls are done from right to left, which
corresponds to the orientation of the marble stack).
\end{proof}

\begin{remark} Considering the domains, we note that marble automata
(transducers without the output) recognize exactly regular languages.  Indeed an
$\SST$ is easily seen to have a regular domain, since it is an extended one-way
automaton.  See \cite{engelfriet1999trips} for another proof of this fact.

\end{remark}

\section{Bounded number of marbles}

\label{sec:bounded}

A natural restriction of our marble transducers is to bound the number of marbles that can be simultaneously present in the stack. Indeed, if a machine uses at most $k$ marbles, it has $\mc{O}(|w|^{k+1})$ possible configurations on input $w$. As a consequence, \emph{it performs its computation in polynomial time}, and the function $f$ it computes is such that $|f(w)|= \mc{O}(|w|^{k+1})$. In particular, the exponential behaviors of Example \ref{ex:expo:mt} are no longer possible.

\begin{definition} A \emph{$k$-marble transducer} is a marble transducer such that every accessible configuration (i.e. reachable from the initial configuration) has a stack of at most $k$ marbles.
\end{definition}

This definition is semantical, but it can easily be described in a syntactical way
by storing the (bounded) number of marbles that are currently dropped on the
input.

\begin{remark} $0$-marble transducers are exactly two-way transducers. 
\end{remark}

As special instances of $1$-marble transducers we get programs with $2$ nested
for loops of shape \verb?for i in {1,...,|w|} ( for j in {1,...,i} (...)  )?.
Indeed, the outer index $i$ corresponds to the marble, and the inner $j$ to the
reading head that cannot move on the right of $i$.  However, this interpretation
does not take the two-way moves into account.

\begin{example} \label{ex:mul}Consider the function $\mul: w \# 0^n \mapsto (w\#)^n$ that produces $n$ copies of $w\#$. It can be realized by a $1$-marble transducer that successively drops the marble from first to last $0$, and each time scans and outputs $w\#$.
\end{example}

\begin{remark} For a $k$-marble transducer, it is enough to have one marble color. Indeed, the colors of the marbles dropped form a finite information that can be encoded in the states.
\end{remark}

The correspondence given by Theorem \ref{theo:sstmt} does not take the number of
marbles into account.  In particular, it does not produce a copyless $\SST$ if
we begin with a $0$-marble transducer.  We shall now give a more precise statement
that relates the maximal number of marbles to the number of "copy layers" in the
$\SST$, as defined below.

\begin{definition} \label{def:ksst}

An $\SST$ $(A,B, Q, \regs,q_0, \iota, \delta,\lambda,F)$ is said to be $k$-\emph{layered} if $\regs$ has a partition of the form $\regs_0, \dots, \regs_k$, such that $\forall q \in Q$, $\forall a \in A$, the following are true:

\item

\begin{itemize}

\item $\forall 0 \le i \le k$, only registers from $\regs_0, \dots, \regs_i$ appear in $\{\lambda(q,a)(x) \mid {x \in \regs_i}\}$;

\item $\forall 0 \le i \le k$, each register $y \in \regs_i$ appears at most once in $\{\lambda(q,a)(x) \mid  x \in \regs_i\}$.

\end{itemize}

\end{definition}

Note that $0$-layered $\SSTs$ are exactly copyless $\SSTs$ (only the second
condition is useful).  For $k \ge 1$, Definition \ref{def:ksst} forces each layer
$\regs_i$ to be "copyless in itself", but it can do many copies of deeper
layers ($\regs_j$ for $j<i$).  This update mechanism is depicted in Figure
\ref{fig:ksst} below; it mainly avoids copying twice a register in itself.

\begin{figure}[h!]

\begin{center}
\begin{tikzpicture}[scale=1]
	\node[] (in) at (-2.5,0) [above,right]{\small Input word};
	\node[above] (in) at (0,-0.2) []{ $b$};
	\node[above] (in) at (1,-0.2) []{  $a$};
	\node[above] (in) at (2,-0.2) []{  $b$};
	\node[above] (in) at (3,-0.2) []{  $b$};
	\node[above] (in) at (4,-0.2) []{  $b$};
	\node[above] (in) at (5,-0.2) []{  $a$};
	\draw[->] (0,-0.5) -- (2.8,-0.5);
	\draw[->] (0,-1) -- (2.8,-1);
	\draw[->] (0,-1.5) -- (2.8,-1.5);
	
	\draw[->,gray,dashed,thick] (3.2,-0.5) -- (3.8,-0.5);
	\draw[->,gray,dashed,thick] (3.2,-1) -- (3.8,-1);
	\draw[->,gray,thick] (3.2,-1) -- (3.8,-0.6);
	\draw[->,gray,dashed,thick] (3.2,-1.5) -- (3.8,-1.5);
	\draw[->,gray,thick] (3.2,-1.5) -- (3.8,-1.1);
	\draw[->,gray,thick] (3.2,-1.5) .. controls (4.5,-3) and (5,-0.5) .. (4.2,-0.5);
	
	\node[] (in) at (3,-0.5) []{\footnotesize $\regs_2$};
	\node[] (in) at (3,-1) []{\footnotesize $\regs_1$};
	\node[] (in) at (3,-1.5) []{\footnotesize $\regs_0$};
	\node[gray] (in) at (4,-0.5) []{\footnotesize $\regs_2$};
	\node[gray] (in) at (4,-1) []{\footnotesize $\regs_1$};
	\node[gray] (in) at (4,-1.5) []{\footnotesize $\regs_0$};
       \end{tikzpicture}
\end{center}
\vspace*{-1.4cm}
\caption{\label{fig:ksst} Update of the registers in a $2$-layered $\SST$}
\vspace*{-0.3cm}
\end{figure}
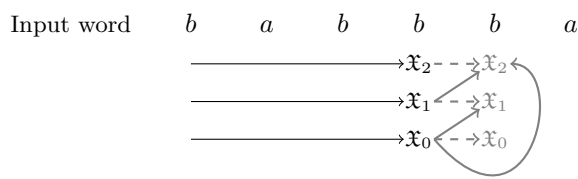

\begin{example} 
  The function $\mul:w \# 0^n \mapsto (w\#)^n$ (Example \ref{ex:mul}) can be
  computed by a 1-layered $\SST$ with $\regs_0 = \{x\}$ and $\regs_1 = \{y\}$ as
  follows.  First, when reading $w\#$, it stores $w\#$ in $x$, while keeping
  $\varepsilon$ in $y$.  Then, each time it sees a $0$, it applies $x \mapsto x,
  y \mapsto xy$.
\end{example}

We now provide a fine-grained correspondence between marbles and registers. Our result indeed extends Theorem \ref{theo:alur}, which corresponds to the case $k=0$.

\begin{theorem} \label{theo:kmtksst} 
  For all $k \ge 0$, $k$-marble transducers and $k$-layered $\SSTs$ describe the
  same class of functions.  The right to left conversion is effective in
  $\PTIME$.
\end{theorem}

\begin{proof}[Proof sketch.] 
  To convert a $k$-layered $\SST$ in a $k$-marble transducer, we adapt the
  transformation of Theorem \ref{theo:sstmt} in order to use no more than $k$
  marbles.  The idea is to write only the recursive calls that correspond to the
  copy of a register, the others being kept implicitly.  The $\PTIME$ complexity
  is obtained by adapting the construction of \cite{dartois2017reversible}.  For
  the converse implication, we first transform the $k$-marble transducer in an
  $\SST$ using \mbox{Theorem \ref{theo:sstmt}}.  Since the function $f$ computed
  by this $\SST$ is such that $|f(w)| = \mc{O}(|w|^{k+1})$, we use \mbox{Lemmas
  \ref{lem:boundedcopy}} and \ref{lem:kbounded} in order to build a $k$-layered
  $\SST$ for $f$.  A large amount of additional work is required to obtain these
  results, and it is the purpose of Section \ref{sec:membership}.
\end{proof}

\section{Membership problems}

\label{sec:membership}

It is clear that a $k$-marble transducer is a particular case of $(k+1)$-marble transducer, which is a particular case of marble transducer (without restrictions). In other words, the classes of functions they define are included in each other. More precisely, these classes describe a strict hierarchy of increasing expressiveness, since $k$-marble transducers can only describe functions such that $|f(w)| = \mc{O}(|w|^{k+1})$ (see Examples \ref{ex:expo:mt} and \ref{ex:hierarchy} for  separation).

\begin{example} \label{ex:hierarchy} 
  The function $\powo^k: a^{n} \mapsto a^{n^k}$ can be computed with $k$
  marbles, but not less (since $|\powo^k(w)|$ is $\Omega(|w|^{k})$).  Let us
  explain the computation of $\powo^2$ with $1$ marble on input $a^n$.  We first
  drop the marble on position $n-1$, go to $1$, and move forward from $1$ to
  $n-1$ while outputting $aa$ at each transition.  Then we lift the marble and
  drop it on $n-2$, and perform the same outputs from $1$ to $n-2$, etc.  At the
  end of this procedure, we have output $a^{2((n-1)+ \cdots + 1)} = a^{n^2 -
  n}$.  It remains to output $a^n$ by reading the input once.
\end{example}

\begin{remark} Generalizing the construction of Example \ref{ex:hierarchy}, we can show that if $P \in \mb{N}[X]$ is a polynomial of degree $k \ge 0$, then $a^n \mapsto a^{P(n)}$ is computable with $k$ marbles (but not less).

\end{remark}

A natural problem when considering a hierarchy is that of \emph{membership}: given a function in some class, does it belong to a smaller one? The objective of this section is to provide a positive answer by showing Theorem \ref{theo:membership} below.

\begin{theorem} \label{theo:membership} Given a function $f$ described by a marble transducer, it is decidable in $\EXP$ whether $f$ can be computed by a $k$-marble transducer for some $k \ge 0$. In that case, we can compute the least possible $k \ge 0$ and build a $k$-marble transducer for $f$.
\end{theorem}

Proposition \ref{prop:growth} below is the key element for the proof, and is also interesting in itself. Indeed, it states that \emph{a polynomial-growth function can always be computed in polynomial time!} In other words, $f$ described by a marble transducer is computable with $k$ marbles if and only if $|f(w)| = \mc{O}(|w|^{k+1})$. Given $f : A^* \rightarrow B^*$, let $|f| : A^* \rightarrow \mb{N}, w \mapsto |f(w)|$.

\begin{definition} \label{def:growth}Let $g : A^* \rightarrow \mb{N}$, we say that $g$ has:
\begin{itemize}
\item \emph{exponential growth}, if $g(w) = \mc{O}(2^{\mc{O}(|w|)})$ and there exists an infinite set $L \subseteq A^*$ such that $g(w)= 2^{\Omega(|w|)}$ when $w \in L$;
\item \emph{$k$-polynomial growth} for $k \ge 0$, if $g(w) = \mc{O}(|w|^k)$ and there exists an infinite set $L \subseteq A^*$ such that $g(w) = \Omega(|w|^k)$ when $w \in L$;
\end{itemize}

\end{definition}

\begin{proposition} \label{prop:growth} Let $f : A^* \rightarrow B^*$ be a total function computed by a marble transducer. Then exactly one of the following is true:
\begin{itemize}
\item $|f|$ has exponential growth, and $f$ is not computable with $k$ marbles for any $k \ge 0$;
\item $|f|$ has $(k+1)$-polynomial growth for some $k \ge 0$, and $f$ is computable with $k$ marbles and $k$ is the least possible number of marbles;
\item $|f|$ has $0$-polynomial growth (i.e. a finite image), and $f$ is computable with $0$ marbles.
\end{itemize}
Moreover, these three properties are decidable in $\EXP$.
\end{proposition}

\begin{remark} If $f$ has a finite image, it is a trivial "step function": $\dom(f)$ is a finite union $ \bigcup_{i} L_i$ of regular languages such that $f$ is constant on each $L_i$.

\end{remark}

The rest of this section is devoted to the proof of these results.  By Theorem
\ref{theo:sstmt}, we first convert our marble transducer in an $\SST$ (in
$\EXP$), and only reason about $\SSTs$ in the sequel.  In fact, considering
$\SSTs$ is crucial: contrary to a marble transducer, an $\SST$ performs only one
pass on its input, which makes it possible to apply pumping-like arguments for
understanding the asymptotic growth of the outputs.  Some proof techniques used
below are inspired from \cite{filiot2017copyful} which only considers deciding
$1$-polynomial growth of $\SSTs$.

\subsection{Simplification of $\SST$}

An $\SST$ is said to be \emph{total} whenever its transition, update and output functions are total. We shall assume that our machine is so. Indeed, it can be completed like a finite automaton, by outputting $\varepsilon$ when out of the domain. Furthermore, this operation does not modify the asymptotic growth of the computed function.

We say that an $\SST$ is \emph{simple} if it is total, it has a single state (i.e. $Q = \{q_0\}$), and its substitutions and output do not use letters (i.e. $\lambda : Q \times A \rightarrow \subst{\regs}{\varnothing}$ and $F:Q \rightarrow \regs^*$). 
To simplify the notations, we write $(A,B, \regs, \iota, \lambda, F)$ for a
simple $\SST$, where $\lambda: A \rightarrow \subst{\regs}{\varnothing}$ and $F
\in \regs^*$.  Indeed, states and the transition function are useless.

\begin{lemma} \label{lem:simple} Given a total $\SST$, we can build an equivalent simple $\SST$ in $\PTIME$.
\end{lemma}

\begin{proof}[Proof.] 
  Let $\trans = (A,B,Q, \regs, q_0, \iota, \delta, \lambda, F)$ be the original
  $\SST$.  We can assume that it uses no letters in the substitutions by storing
  them in constant registers (using the initial function).  To remove the
  states, we let $\regs' := Q \times \regs$ be our new register set.  In our new
  machine, register $(q,x)$ will contain the value of $x$ if $q$ is the current
  state of $\trans$, and $\varepsilon$ otherwise.  The update function $\lambda'
  : A \rightarrow \subst{\regs'}{B}$ and output $F' \in \regs'^*$ are defined as
  follows:
  $$
  \lambda'(a)(q,x) = \prod_{p \mid \delta(p,a) = q} \mu_{p}(\lambda(p,a)(x)) \text{ and } F' =  \prod_{q \in \dom(F)} \mu_{p}(F(q))
  $$
  where $\mu_p$ replaces $y \in \regs$ by $(p,y) \in \regs'$. During a run, at most one term of the concatenation defining $\lambda'(a)(q,x)$ is nonempty, the one which corresponds to the true state $p$ of $\trans$. 
\end{proof}

\begin{remark} However, this construction does not preserve copylessness nor $k$-layeredness.
\end{remark}

\subsection{Asymptotic behavior of $\mb{N}$-automata and $\SSTs$}

Given a simple $\SST$, we first build an $\mb{N}$-automaton that "computes" the
size of the words stored in the registers along a run of the $\SST$.  As we
shall see, the growth of functions produced by $\mb{N}$-automata exactly matches
the case disjunction of Proposition \ref{prop:growth}.

\begin{definition} \label{def:naut} An $\mb{N}$-automaton $\auto = (A,\staa,\inia,\weia,\fina)$ consists in:

\item

\begin{itemize}

\item an input alphabet $A$;

\item a finite set $\staa$ of states;

\item an initial row vector $\inia \in \mb{N}^{\staa}$ and a final column vector $\fina \in \mb{N}^{\staa}$;

\item a monoid morphism $\weia : A^* \rightarrow \mb{N}^{\staa \times \staa}$ (weight function).

\end{itemize}

\end{definition}

The automaton $\auto$ computes the total function $ A^* \rightarrow \mb{N}, w \mapsto \inia \weia(w) \fina$. We say that it is \emph{trim} if $\forall q \in \staa$, $\exists u,v \in A^*$ such that $(\inia \weia(u))(q) \ge 1$ and $(\weia(v)\fina)(q) \ge 1$.

Let $ \trans =(A,B, \regs, \iota, \lambda, F)$ be a simple $\SST$, we define its
\emph{flow automaton} $\fauto{\trans} := (A, \regs, \inia,\weia,\fina)$ as the
$\mb{N}$-automaton with input $A$, states $\regs$, and:
\begin{itemize}
\item for all $x \in \regs$, $\inia(x) = |\iota(x)|$ (number of letters initialized in $x$);

\item for all $x \in \regs$, $\fina(x)$ is the number of occurrences of $x$ in $F$;

\item for all $a \in A, x,x' \in \regs$, $\weia(a)(x,x')$ is the number of occurrences of $x$ in $\lambda(a)(x')$.
\end{itemize}

Recall that $\trans^{w}(x)$ is "the value of $x$ after reading $w$ in $\trans$"; the flow automaton indeed computes the size of these values. We get the following by induction.

\begin{claim} \label{claim:nauto} For all $w \in A^*$ and $x \in \regs$, we have $(\inia \weia(w))(x) = |\trans^{w}(x)|$. In particular, if $f$ is the function computed by $\trans$, then $\fauto{\trans}$ computes $|f|$.
\end{claim}

Without loss of generality, we can assume that $\fauto{\trans}$ is trim. Indeed, if $x \in \regs$ is such that $(\inia \weia(u))(x) = 0$ for all $u \in A^*$, then $x$ always has  value $\varepsilon$ and can be erased everywhere in $\trans$. Similarly, if $(\weia(v)\fina)(x) = 0$ for all $v \in A^*$, $x$ is never used in the output.

Let us now study in detail the asymptotic behavior of $\mb{N}$-automata.

\begin{lemma} \label{lem:graph} Let $\auto = (A, \staa, \inia, \weia,\fina)$ be a trim $\mb{N}$-automaton that computes a function $g:A^* \rightarrow \mb{N}$. Then one of the following holds, and it can be decided in $\PTIME$:

\item 

\begin{itemize}

\item \label{po:bound} $g$ has an exponential growth;

\item \label{po:partition} $g$ has $k$-polynomial growth for some $k \ge 0$ and $\staa = \biguplus_{0 \le i \le k} S_i $ is such that:
\begin{itemize}

\item $\forall q,q' \in \staa$, if $\exists w \in A^*$ such that 
$\weia(w)(q,q') \ge 1$ then $q \in S_i, q'\in S_{j}$ for some $i \le j$;

\item $\exists B \ge 0$ such that $\forall 0 \le i \le k$, $\forall q,q' \in S_i, \forall w \in A^*$, $\weia(w)(q,q') \le B$;

\item $\forall q \in S_i, (\inia \weia(w))(q) = \mc{O}(|w|^i)$.

\end{itemize}

\end{itemize}

Furthermore, $k$ and $S_0, \dots, S_k$ are computable in $\PTIME$.

\end{lemma}

\begin{remark} An upper bound $B$ can be described explicitly, see e.g. \cite{sakarovitch2008decidability}.

\end{remark}

\begin{proof}[Proof sktech.] Very similar results are obtained in \cite{weber1991degree} for computing ambiguity of finite automata (which corresponds to $\mb{N}$-automata with weights in $\{0,1\}$ only).  However, in order to keep the paper self-contained, we give a detailed proof  in Appendix.

Mainly, we look for the presence of the two patterns from \cite{sakarovitch2008decidability} in the weights of $\auto$:
\begin{itemize}

\item \emph{heavy cycles} ($\exists q \in \staa, v \in A^*$ such that $\weia(v)(q,q) \ge 2$), that creates exponential growth ;

\item  \emph{barbells} ($\exists q \neq q', v \in A^{+}$, such that $\weia(v)(q,q)  \ge 1$, $\weia(v)(q,q')  \ge 1$ and $\weia(v)(q',q')  \ge 1$) such that a chain of $k$ barbells induces $k$-polynomial growth.
\end{itemize}
\vspace*{-1\baselineskip}
\end{proof}

As a consequence, if $f$ is computed by an $\SST$, then $|f|$ has either exponential growth or $k$-polynomial growth for some $k \ge 0$. Furthermore, we can decide it in $\PTIME$.

It remains to show that if $|f|$ has a $(k+1)$-polynomial growth, then $f$ is computable by a $k$-layered $\SST$. For this, we shall use the partition of Lemma \ref{lem:graph} that splits the simple $\SST$ (via the states of its flow automaton) in a somehow $k$-layered way. However, the layers obtained are not directly copyless, and another transformation is necessary.

\subsection{Construction of $k$-layered $\SST$ in the polynomial case}

If $|f|$ has $(k+1)$-polynomial growth, then Lemma \ref{lem:graph} partitions
the simple $\SST$ in $k+2$ layers.  Our first concern is to get $k+1$ layers
only, since we want a $k$-layered $\SST$.  In the next definition,
$\lambda(p,w)$ denotes the substitution applied when reading $w \in A^*$ from
$p\in Q$, that is $\lambda(p,w[1]) \circ \cdots
\lambda(\delta(p,w[1{:}(i-1)]),w[i]) \circ \cdots \circ
\lambda(\delta(p,w[1{:}(|w|-1)]),w[|w|])$.

\begin{definition} \label{def:bounded} 
  We say that an $\SST$ $(A, B, Q, \regs, q_0, \iota, \delta, \lambda, F)$ is
  \emph{$(k,B)$-bounded} if there exists a partition $\regs_0,\regs_1, \dots,
  \regs_k$ of $\regs $ such that $\forall q \in Q, a \in A, w \in A^*$:
  \begin{itemize}
    \item $\forall 0 \le i \le k$, only registers from $\regs_0, \dots, \regs_i$ appear in $\{\lambda(q,a)(x) \mid {x \in \regs_i}\}$;
    \item $\forall 0 \le i \le k$, each $y \in \regs_i$ appears at most
    $B$ times in $\{\lambda(q,w)(x) \mid x \in \regs_i\}$.
  \end{itemize}
\end{definition}

For $k=0$, Definition \ref{def:bounded} means that at most $B$ copies of $y$ are "useful" during a run.  The $(0,B)$-bounded $\SSTs$ are known as \emph{$B$-bounded (copy) $\SSTs$}
in \cite{dartois2016aperiodic} (however, contrary to what is said in
\cite{dartois2016aperiodic,filiot2017copyful}, it is not the same definition as
the "bounded copy" of \cite{alur2012regular}). For some $k \ge 1$, a $(k,B)$-bounded $\SST$ is similar to a $k$-layered $\SST$, except that each layer is no longer "copyless in itself" but "$B$-bounded in itself". In particular, we note that $(k,1)$-bounded $\SSTs$ exactly correspond to $k$-layered $\SSTs$.

\begin{lemma} \label{lem:boundedcopy} 
  Given a simple $\SST$ that computes a function $f$ such that $|f|$ has
  $(k+1)$-polynomial growth, we can build an equivalent $(k,B)$-bounded $\SST$ for
  some $B \ge 0$.
\end{lemma}

\begin{proof}[Proof sketch.] Let $S_0, \dots, S_{k+1}$ be the partition of the registers given by Lemma \ref{lem:graph}. During a run, note that the registers in $S_0$ can only store strings of a bounded size. The idea is to remove $S_0$ and hardcode the content of each $x \in S_0$ in  a finite set of states. The transition function is defined following their former updates. The new update function is defined by replacing the mention of $x \in S_0$ by its explicit content (given by the current state).
\end{proof}

\begin{remark} As for weighted automata above, an upper bound $B$ can effectively be computed.
\end{remark}

It is known that a $(0,B)$-bounded $\SST$ can be transformed in a copyless $\SST$. The proof requires rather complex constructions, that we generalize for a $(k,B)$-bounded $\SST$.

\begin{lemma} \label{lem:kbounded} Given a $(k,B)$-bounded $\SST$, we can build an equivalent $k$-layered $\SST$.
\end{lemma}

\begin{proof}[Proof sketch.] The proof is done by induction on $k \ge 0$. Indeed a $(k,B)$-bounded (resp. $k$-layered) $\SST$ is somehow a $B$-bounded (resp. copyless) $\SST$, that can also "call" registers from the lower layers. Hence we only need to focus on transforming \emph{one} layer. The difficulty is to take into account copies coming from the lower layers. This is done by introducing an intermediate model of \emph{$\SST$ with external functions} ($\SSTF$), which corresponds to an $\SST$ with a set of functions $\oras$ that can be called in an oracle-like style.

Thus our proof roughly consists in showing that a $B$-bounded $\SSTF$ (in the sense of Definition \ref{def:bounded}) can be transformed in a copyless $\SSTF$. This is done in two steps. First, we transform the $B$-bounded $\SSTF$ in a copyless \emph{non-deterministic} $\SSTF$, following the ideas of \cite{dartois2016aperiodic} for $\SST$. Non-deterministic transducers usually compute relations between words, but we in fact obtain an \emph{unambiguous} machine (i.e. that has at most one accepting run on each input), hence describing a function. Second, we show that a copyless unambiguous non-deterministic $\SSTF$ can be converted in a copyless $\SSTF$. This transformation relies on the techniques of \cite{alur2012regular} (developped for $\SST$ over infinite words).
\end{proof}

\begin{proof}[Proof of Theorem \ref{theo:membership}.] 
We first transform an $\SST$ into a simple $\SST$ (Lemma \ref{lem:simple}) and build its flow automaton. Using this machine, one can decide what is the growth of $|f|$ (Lemma \ref{lem:graph}). If $|f|$ has $(k+1)$-polynomial growth, we then build a $(k,B)$-bounded $\SST$ that computes it (Lemma \ref{lem:boundedcopy}) and finally a $k$-layered $\SST$ (Lemma \ref{lem:kbounded}).
\end{proof}

\section{Conclusion and outlook}

We presented in this paper a new correspondence between $\SSTs$ and marble
transducers.  Showing that two models are equivalent is always interesting in
itself, but our result also provides a deeper understanding of their behaviors.
Indeed, it relates recursive and iterative programs (marbles) to streaming
algorithms ($\SSTs$), which are suitable for program optimization problems.
Since the equivalence problem is decidable for $\SSTs$ \cite{filiot2017copyful},
we also obtain for free that it is the case for marble transducers (which was not
previously known).

Note that our model is not closed under composition. It is the case for obvious asymptotic growth reasons, since marble transducers can compute one exponential ($\pow : a^n \mapsto a^{2^n}$) but not two of them ($\pow \circ \pow$). More surprisingly, there exist polynomial-size compositions that cannot be expressed by our transducers, as shown below.

\begin{claim} $\mul : w \#0^n \mapsto (w\#)^n$ is computable by an $\SST$, but not $0^n\#w \mapsto (w\#)^n$.
\end{claim}
This result mainly comes because marbles and $\SSTs$ give an \emph{orientation} on the input: we stack marbles "on the right", and the $\SST$ is a streaming process "from left to right".

\subparagraph*{Marbles and pebbles.} A non-oriented generalization of $k$-marbles, named \emph{$k$-pebble transducers} \cite{bojanczyk2018polyregular}, has recently been investigated in detail. In this case, the reading head is allowed to move on the right of a mark without lifting it, while keeping a stack discipline. The typical example of function computable with $1$ pebble is $\sq : A \rightarrow A \uplus \{\overline{a}\mid a \in A\}, abc \mapsto \overline{a}bc a \overline{b}c ab \overline{c}$, which associates to $w$ the concatenation with $|w|$ copies of itself, the $i$-th copy having its $i$-th letter overlined. This function cannot be computed by a marble transducer.

In \cite{lhote2020pebble}, the membership problem is solved for the classes of $k$-pebble transducers. Despite their similarity, neither their result (Theorem \ref{theo:nathan} below) nor our \mbox{Proposition \ref{prop:growth}} imply each other, and the proof techniques are significantly different. Indeed we consider different classes of functions, ours being designed for streaming implementations, but not theirs. The relationship between marbles and pebbles is depicted in Figure \ref{fig:conclu}.

\begin{theorem}[\cite{lhote2020pebble}] \label{theo:nathan} A function $f$ described by a $k$-pebble transducer is computable by an $\ell$-pebble transducer if and only if $|f(w)| = \mc{O}(|w|^{\ell+1})$ (and this property is decidable).
\end{theorem}

Contrary to us, they do not obtain tight asymptotic bounds. Furthermore, they only consider machines with a bounded number of marks, i.e. no exponential growths.

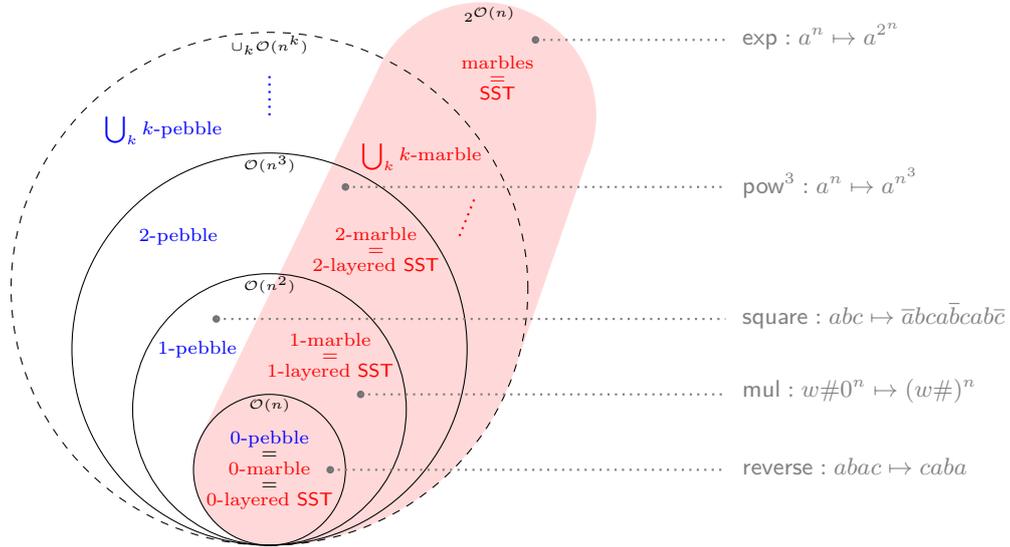
\begin{figure}[h!]

    \def\acircle{(0,0) circle (1)}
    \def\bcircle{(0,0.8) circle (1.8)}
    \def\ccircle{(0,1.6) circle (2.6)}
    \def\mcircle{(6.32,-1.5) circle (7.5)}
    \def\icircle{(0,2.4) circle (3.4)}
    
    	\begin{center}
	\hspace*{-0.4cm}
        \begin{tikzpicture}{scale=0.9}

	   % Remplissage rose
            \fill[red!15] \acircle;     
            \fill[red!15] {(2.8,4.7) circle (1.5)};

            \begin{scope}
            \clip {(0,-1) rectangle (8,2)};
            \fill[red!15] \icircle;
            \end{scope}
            
            \fill[red!15] (-0.94,0.37) -- (1.404,5.25) -- (4.246,4.3) -- (3.25,1.4) -- cycle;

            \draw \acircle node {$\substack{\textcolor{blue}{0\text{-pebble}} \\ =\\ \textcolor{red}{0\text{-marble}} \\ = \\ \textcolor{red}{0\text{-layered } \SST}}$};
            \draw \bcircle ;
            \draw \ccircle;
            \draw[dashed] \icircle;
            
            \draw[blue,dotted,thick] (0,4.7) -- (0,5.25);
             \draw[red,dotted,thick] (2.5,3.1) -- (2.7,3.6);
            
            \node at (0.8,1.5) {\textcolor{red}{$\substack{1\text{-marble} \\ = \\ 1\text{-layered } \SST}$}};
            \node at (-0.95,1.6) {$\substack{\textcolor{blue}{1\text{-pebble}}}$};
            
             \node at (0,0.86) {\tiny $\mc{O}(n)$};
             \node at (0,2.45) {\tiny $\mc{O}(n^2)$};
             \node at (0,4.05) {\tiny $\mc{O}(n^3)$};
             \node at (0,5.63) {\tiny $\cup_k \mc{O}(n^k)$};
            
            \node at (1.4,2.9) {\textcolor{red}{$\substack{2\text{-marble} \\ = \\ 2\text{-layered } \SST}$}};
            \node at (-1.2,3.1) {$\substack{\textcolor{blue}{2\text{-pebble}}}$};
            \node at (-1.4,4.5) {$\substack{\textcolor{blue}{\bigcup_{k} k\text{-pebble}}}$};
            \node at (2,4.15) {$\substack{\textcolor{red}{\bigcup_{k} k\text{-marble}}}$};

            \node at (3,5.2) {\textcolor{red}{$\substack{\text{marbles} \\ = \\ \SST}$}};
             \node at (2.9,6.04) {\tiny $2^{\mc{O}(n)}$};

             \fill[gray!110] {(0.8,0) circle (0.05)};
             \draw[gray,dotted,thick] (0.8,0) -- (6,0);
             \node[right] at (6.1,0.05) {\small \textcolor{gray}{$\mirror: abac \mapsto caba$}};
             
             \fill[gray!110] {(1.2,1) circle (0.05)};
             \draw[gray,dotted,thick] (1.2,1) -- (6,1);
             \node[right] at (6.1,1.05) {\small \textcolor{gray}{$\mul: w\#0^n \mapsto (w\#)^n$}};

             \fill[gray!110] {(-0.7,2) circle (0.05)};
             \draw[gray,dotted,thick] (-0.7,2) -- (6,2);
             \node[right] at (6.1,2.05) {\small \textcolor{gray}{$\sq: abc \mapsto \overline{a}bc a \overline{b}c ab \overline{c}$}};

            \fill[gray!110] {(1,3.75) circle (0.05)};
             \draw[gray,dotted,thick] (1,3.75) -- (6,3.75);
             \node[right] at (6.1,3.8) {\small \textcolor{gray}{$\powo^3: a^n \mapsto a^{n^3}$}};
             
             \fill[gray!110] {(3.5,5.7) circle (0.05)};
             \draw[gray,dotted,thick] (3.5,5.7) -- (6,5.7);
             \node[right] at (6.1,5.75) {\small \textcolor{gray}{$\pow: a^n \mapsto a^{2^n}$}};

        \end{tikzpicture}
        \end{center}
        
    \caption{\label{fig:conclu} Classes of functions studied in this paper (\textcolor{red}{red}) and in \cite{lhote2020pebble}{} (\textcolor{blue}{blue}).}
\end{figure}

\subparagraph*{Future work.} Our work opens the way to a finer study of the classes of functions described by marble and pebble transducers. The \emph{membership problem from $k$-pebble to $k$-marble} is worth being studied to complete the decidability picture. It reformulates as follows: given a function computed by a pebble transducer, can we implement it in streaming way? The answer seems to rely on combinatorial properties of the output. Another perspective is to define a \emph{logical description} of our transducers, which allows to specify their behavior in a non-operational fashion. No formalism is known for marble transducers, but it is known since \cite{engelfriet2001mso} that two-way transducers correspond to \emph{monadic-second-order transductions}. On the other hand, \cite{droste2019aperiodic} studies in detail a \emph{weighted logics} which describes the functions computed by weighted automata. Using proof techniques which are similar to ours, they relate the asymptotic growth of the function to the logical connectors required to describe it.

%%
%% Bibliography
%%

\newpage

\bibliographystyle{alpha}% the recommended bibstyle
\bibliography{marbles}

\newpage

%%
%% Appendix
%%

%\tableofcontents

\appendix

\section{Proof of Theorem \ref{theo:sstmt}}

\subsection{From marble transducers to $\SSTs$}

We transform a marble transducer into an $\SST$.  The main idea is to keep track
of the right-to-right behaviors (the "crossing sequence") of the prefix read so
far.  We adapt the classical transformation of two-way automata to one-way
automata \cite{shepherdson1959reduction}.

Consider a marble transducer $\trans = (A,B,Q,C,q_0, \delta, \lambda,F)$ on input $\lmark w \rmark$. We denote by $\rightarrow$ its transition relation (see page \pageref{mark:transrel}). When reaching a position $m$ of an input $w$, the $\SST$ will keep track the following information (see Figure \ref{fig:cross}):
\begin{itemize}
  \item the state $\first_m \in Q \uplus \{\bot\}$ that is "the state of
  $\trans$ the first time it reaches position $m+1$".  More formally, $\first_m$
  is the state such that $(0,q_0, \vide) \rightarrow^* (m+1,\first_m,
  \vide)$ for the first time, and we use $\bot$ if it does not exist.
  Since ${\first_m}$ is a bounded information, it is coded in the state of the
  $\SST$.  We also store in a register the concatenation $\lambda_{\first_m}$ of
  the outputs $\lambda$ along this run $(0,q_0, \vide) \rightarrow^*
  (m+1,\first_m, \vide)$;

  \item a function $\nnext_m(q) \colon Q \rightarrow Q \uplus \{\bot\}$, which
  gives for each $q \in Q$, the state such that $(m,q, \vide)
  \rightarrow^* (m+1,\nnext_m(q), \vide)$ for the first time (with $\bot$
  if it does not exist).  Note that for any stack $\pi$ with no marbles dropped
  on $w[1{:}m]$, $\nnext_m(q)$ is also the state such that $(m,q, \pi)
  \rightarrow^* (m+1,\nnext_m(q), \pi)$ for the first time (and this is a
  "similar" run, with the same transitions and the same output).  This bounded
  information is coded in the state of the $\SST$.  In a register, we store the
  concatenation $\lambda_{\nnext_m(q)}$ of the outputs along the run.

\end{itemize}

\begin{remark} The marble stack is necessarily empty at the first visit of a position.
\end{remark}

\begin{figure}[h!]

\begin{center}
\begin{tikzpicture}[scale=1]
	\node (in) at (-3,0) [above,right]{ Input word};
	
	\draw[fill = red!20,dashed](4.5,0.5) rectangle (5.5,-2.5);
		
	\node (in) at (0,-0.5) []{\footnotesize  $q_0$};
	\draw[->,dotted,thick](0.2,-0.5) to (5.8,-0.5);
	\node (in) at (5.8,-0.5) [right]{\footnotesize $ \first_m$};

	\node (in) at (5,-1) []{\footnotesize  $q_{1}$};
	\node (in) at (5.8,-1.5) [right]{\footnotesize  $\nnext_m(q_1)$};
	\draw[->,dotted,thick] (4.8,-1) .. controls (1,-1) and (1,-1.5) .. (5.8,-1.5) ;
	
	\node (in) at (5,-2) []{\footnotesize  $q_{2}$};
	\node (in) at (5.8,-2) [right]{\footnotesize  $ \nnext_m(q_2) =  \bot$};
	\draw[->,dotted,thick](4.8,-2) to (3,-2);
	\node (in) at (2.3,-2) []{\footnotesize  deadlock};	

	\node[above] (in) at (0,-0.2) []{  $\lmark$};
	\node[above] (in) at (1,-0.2) []{  $b$};
	\node[above] (in) at (2,-0.2) []{  $a$};
	\node[above] (in) at (3,-0.2) []{  $b$};
	\node[above] (in) at (4,-0.2) []{  $b$};
	\node[above] (in) at (5,-0.2) []{  $b$};
	\node[above] (in) at (6,-0.2) []{  $a$};
	\node[above] (in) at (7,-0.2) []{  $\rmark$};
       \end{tikzpicture}
\end{center}
       \caption{\label{fig:cross} Crossing sequences in a marble transducer}
\end{figure}
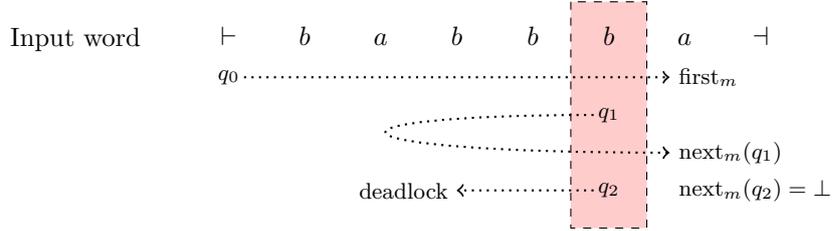

\begin{remark} A $\bot$ can be used for two possible reasons: either the marble transducer is blocked before coming to $m+1$, or it enters an infinite loop on the prefix $\lmark w[1{:}m]$.
\end{remark}

\subparagraph*{Updates of the $\SST$.} We have to show that the $\SST$ can update
this abstraction of the behavior of $\trans$.  Assume that $\nnext_m(q)$ and
$\lambda_{\nnext_m(q)}$ are known for all $q \in Q$, we want to determine
$\nnext_{m+1}$ (the case of $\first_m$ is very similar).  Let $a = w[m+1]$ and
$f = \nnext_{m} \colon Q \rightarrow Q \uplus \{\bot\}$, we consider the set of
functions $Q \rightarrow Q \uplus \{\bot\}$ ordered on the images by the flat
ordering on $Q$ and $\bot \le q$ for all $q \in Q$.  We define the functions
$(g_c)_{c \in C \uplus \{\varnothing\}}$ as the the least fixed point of the
following equations:

$$ 
g_{\varnothing}(q) = \left\{
	\begin{array}{ll}
		q' & \text{if } \delta(q,a, \varnothing) = (q', \rmove) \\
		g_{\varnothing}(f(q')) & \text{if } \delta(q,a, \varnothing) = (q', \lmove)\\
		g_{c}(q') & \text{if } \delta(q,a, \varnothing) = (q', \drop_c)\\
	\end{array}
	\right.
$$

$$ 
g_{c}(q) = \left\{
	\begin{array}{ll}
		g_{\varnothing}(q') & \text{if } \delta(q,a, c) = (q', \lift)\\
		g_{c}(f(q')) & \text{if } \delta(q,a, c) = (q', \lmove)\\
	\end{array}
	\right.
$$

As described in the example below, these equations describe how the former
$\nnext_m$ is "stitched" with the moves performed on $w[m]$, in order to compute
$\nnext_{m+1}$.  The fixpoint can be computed by a saturation algorithm in
$\PTIME$ with respect to $|Q|$.

\begin{example} \label{ex:sew} 
  We suppose that that $\delta(q,a,\varnothing) = (q_1, \drop_c)$,
  $\delta(q_1,a,c) = (q_2, \lmove)$, $f(q_2) \neq \bot$, $\delta(f(q_2),a,c) =
  (q_3, \lift)$, $\delta(q_3,a,\varnothing) = (q_2, \lmove)$ and
  $\delta(f(q_2),a,\varnothing) = (q', \rmove)$.  Then, we have
  $q'=g_{\varnothing}(q)=g_c(q_1)=g_c(f(q_2))=g_{\varnothing}(q_3)=g_{\varnothing}(f(q_2))$.
  A possible run from $(q,m+1, \vide) $ to the first visit of position $m+2$ is
  depicted in Figure \ref{fig:sew}.
\end{example}

\begin{figure}[h!]

\begin{center}
\begin{tikzpicture}[scale=1]

	\node (q) at (0,-0.5) []{\footnotesize  $(q,\varnothing)$};
	\node (q1) at (0,-1.5) []{\footnotesize  $(q_1,c)$};
	\draw[->,thick](q) to (q1);
	
	\node (q2) at (-2,-1.5) {\footnotesize $(q_2,\varnothing)$};
	\draw[->,thick](q1) to (q2);
	\node (fq2) at (0,-2.5) {\footnotesize $(f(q_2),c)$};
	\draw[->,dotted,thick] (q2) .. controls (-5,-1.75) and (-5,-2.25) .. (fq2) ;
	
	\node (q3) at (0,-3.5) {\footnotesize $(q_3,\varnothing)$};
	\draw[->,thick](fq2) to (q3);	
	
	\node (q2b) at (-2,-3.5) {\footnotesize $(q_2,\varnothing)$};
	\draw[->,thick](q3) to (q2b);
	\node (fq2b) at (0,-4.5) {\footnotesize $(f(q_2),\varnothing)$};
	\draw[->,dotted,thick] (q2b) .. controls (-5,-3.75) and (-5,-4.25) .. (fq2b) ;
	
	\node (g) at (2.5,-4.5) {\footnotesize $(g_{\varnothing}(q),\varnothing)$};
	\draw[->,thick](fq2b) to (g);
	
	\node[above] (in) at (-3.5,-0.2) []{\footnotesize  $\cdots$};	
	\node[above] (in) at (-2,-0.2) []{\footnotesize  $w[m]$};
	\node[above] (in) at (0,-0.2) []{\footnotesize  $w[m+1]$};
	\node[above] (in) at (1.8,-0.2) []{\footnotesize  $\cdots$};

       \end{tikzpicture}
\end{center}
       \caption{\label{fig:sew} Example of run starting from configuration $(q,m+1,\vide)$, where $f = \nnext_m$}
\end{figure}
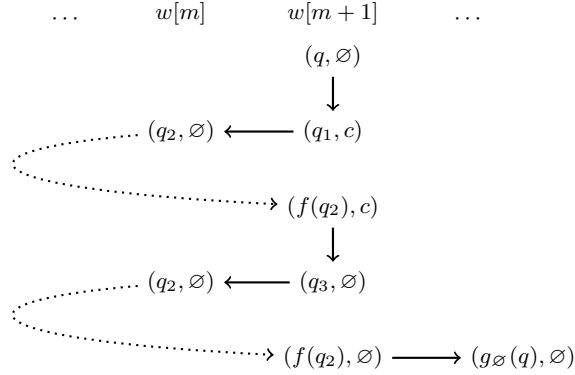

\begin{claim} \label{claim:sew} Let $q \in Q$ and $c \in C$. Then:
\begin{itemize}
\item if $g_c(q) = q'\in Q$, then $(m+1,q, (c,m+1)) \rightarrow^* (m+2,q', \vide)$ first visit of  $m+2$;
\item if $g_c(q) = \bot \in Q$, then the run starting in $(m+1,q, (c,m+1))$ never visits position $m+2$.
\end{itemize}
The same holds when replacing $c$ with $\varnothing$, in words $\nnext_q = g_{\varnothing}$.
\end{claim}

The transitions of the $\SST$ are defined by hardcoding the solution of the
equations.  As noted in Example \ref{ex:sew} above, the computation of the
fixpoint also provides a description of the run $(m+1,q, \vide) \rightarrow^*
(m+2,\nnext_{m+1}(q), \vide)$.  Using this description, we can construct a
substitution that describes $\lambda_{\nnext_{m+1}(q)}$ in terms of
$\lambda_{\nnext_{m}(q)}$.

\begin{example} Following Example \ref{ex:sew}, we have:
  $$
  \lambda_{\nnext_{m+1}(q)} = \lambda(q,a,\varnothing)  \lambda(q_1,a,c) \lambda_{\nnext_{m}(q_2)} \lambda(f(q_2),a,c)  \lambda(q_3,a,\varnothing ) \lambda_{\nnext_{m}(q_2)}  \lambda(f(q_2),a,\varnothing). 
  $$
  The values $\lambda(q,a,\varnothing)$ are constants which will be hardcoded in
  the substitutions.  Note that the substitution described above uses two copies
  of the register $\lambda_{\nnext_{m}(q_2)}$.  Indeed, the run depicted in
  Figure \ref{fig:sew} uses twice the same path starting from $q_2$ in position
  $m$.  Such a situation cannot occur with a two-way transducer, since it would
  induce a loop (it is not the case here because of the marble $c$).
\end{example}

\subparagraph*{Output function.} When the whole word is read, we can recombine
all pieces of information in order to obtain the output of the marble transducer
(when it accepts).  The construction is similar to that of the update, by
stitching the different pieces of the run.

\subparagraph*{Complexity of the construction.} Due to the use of functions $\nnext_{m} \colon Q \rightarrow Q \uplus \{\bot\}$, the $\SST$ has a number of states and transitions which is exponential in $|Q|$. Given two states, the existence of a transition between them and the computation of its substitution can be done in $\PTIME$, hence the whole construction can be performed in $\EXP$.

\subsection{From $\SSTs$ to marble transducers}

\label{subs:proof:sstmt}

Consider an $\SST$ $\trans = (A,B,Q,\regs, q_0, \iota, \delta, \lambda, F)$
computing a function $f$.  We assume that $\delta$, $\lambda$ and $F$ are total
functions.  Indeed, we can complete them and treat the domain of $f$ separately
(it is a regular language).  The main idea is to execute a simple recursive
algorithm for $f$, then we show that it can be implemented with a marble transducer.

\subparagraph*{Recursive algorithm.} Given a word $w$, $0 \le m \le |w|$ and $x
\in \regs$, Algorithm \ref{algo:sstmt} computes $\trans^{w[1{:}m]}(x)$ (that is
"the value stored in $x$ after $\trans$ has read $w[1{:}m]$", see page
\pageref{mark:valu}).  For this, it finds the substitution $x \mapsto \alpha$
that was applied at $m$, and then makes recursive calls to compute the values of
the registers appearing in $\alpha$, at position $m-1$.  The claim below follows
after an easy induction.

\begin{algorithm}[h!]
\SetKw{KwVar}{Variables:}
\SetKwProg{Fn}{Function}{}{}
\SetKw{In}{in}
\SetKw{Out}{Output}

 \Fn{$\val (x,m,w)$}{
 
 	\tcc{$x \in \regs$ register to be computed, $0 \le m  \le |w|$ current position}
 
	\eIf{$m = 0$}{
		
		\KwRet{$\iota(x)$};
		\tcc{Initialization of the registers}}{

		$q \leftarrow \delta(q_0,w[1{:}(m-1)])$;
		\tcc{State $q$ before reading $w[m]$}
				
		$\alpha \leftarrow \lambda(q,w[m])(x)$;
		\tcc{Current substitution $x \mapsto \alpha$}

		$v \leftarrow \vide$; 	\tcc{Will store the value of $x$}

		\For{$i$ \In $\{1, \dots, |\alpha|\}$}{
		
			\eIf{$\alpha[i] \in B$}{
			
			$v \leftarrow v \cdot \alpha[i]$;
			\tcc{Letter $\alpha[i] \in B$ added to $x$}
			
			}{
			
			$v \leftarrow v \cdot \val(\alpha[i],m-1,w)$;
			
			\tcc{Compute recursively the value of $\alpha[i] \in \regs $ at $m-1$}
			
			}
			
		}
		
		\KwRet{$v$};
		\tcc{Value of $x$ is output}

	}
	
}
	
 \caption{\label{algo:sstmt} Computing the value of $x \in \regs$ at position $m$ of $w$}
\end{algorithm}

\begin{claim} If $x \in \regs$ and $0 \le m \le |w|$, $\val(x,m,w)$ computes $\trans^{w[1:m]}(x)$.
\end{claim}

\subparagraph*{Implementation by a marble transducer.} We show how Algorithm \ref{algo:sstmt} can be implemented with a marble transducer. First, let us explain how to compute the state $q = \delta(q_0,w[1{:}(m-1)])$ each time we need it. We drop a special marble $\bullet$ in the current position $m-1$. Then, we move to the left symbol $\lmark$. Finally, we simulate the transitions of $\trans$ from position $1$ to position $m-1$ (that is recovered thanks to marble $\bullet$) and finally we lift the marble $\bullet$.

We now deal with the recursive execution of the algorithm. The main idea is to use the marbles in order to write explicitly the recursivity stack of $\val$ on the word. Given $\alpha \in (B \uplus \regs)^*$ we define $\mar(\alpha) \subseteq (B \uplus \regs \uplus \{\overline{x}\mid x \in \regs\})^*$ to be the set of copies of $\alpha$ in which exactly one register is overlined. 

\begin{example}

If $\regs = \{x,y\}$ and $B = \{b\}$, $\mar(xbybx) = \{\overline{x}bybx, xb\overline{y}bx\, xbyb\overline{x} \}$.

\end{example}

The marble transducer has marble colors $C := \{\bullet\} \biguplus_{q,a,x}
\mar(\lambda(q,a)(x)) $.  When computing the value of $x$ at position $m$, it
will move on the prefix $\lmark w[1{:}m]$ to output $\val(x,m,w)$.  How?  First,
it gets $q$ as shown before, then $\alpha$.  Then, it performs the (hardcoded)
"for" loop reading $\alpha$.  When it sees a letter $\alpha[i] \in B$, it
outputs it.  When it sees a register $\alpha[i] \in \regs$, it drops the marble
$\alpha[1{:}i-1] \overline{\alpha[i]} \alpha[i+1{:}|\alpha|]$ on the current
position $m$ and moves left to compute $\alpha[i]$.  Once $\alpha[i]$ at $m-1$
is recursively computed, the transducer moves right.  Thanks to the marble
there, it remembers that it was computing index $i$ of $x \mapsto \alpha$ and
pursues the loop.

The case when $m = 0$ is detected by reading the letter $\lmark$: here the
transducer does not go left but outputs $\iota(x)$ instead.  To compute the
final output, it starts from $\rmark$, computes $F(q_{|w|})$ using $\bullet$ as
shown above, and applies the previous backward algorithm.

Note that the stack policy is respected here, because to compute the value of a
register at position $m$, the marble transducer only needs to visit positions on
the left of $m$.

\subparagraph*{Complexity of the construction.} The marble transducer can be
constructed in $\PTIME$ from $\trans$ (with respect to $|Q|$ and $\sum_{q,a,x} |\lambda(q,a)(x)|$). Indeed, its set of marbles has size $\sum_{q,a,x} |\mar(\lambda(q,a)(x))| + 1$. When executing the recursive algorithm, we only need to store $q \in Q$, $\alpha \in \subst{\regs}{B}$ and the position $1 \le i \le |\alpha|$ in the state of the marble transducer. Hence the set of states and transitions can clearly be described in $\PTIME$. The part of the machine designed to compute $q$ using $\bullet$ is also easy to describe.

\section{Proof of Theorem \ref{theo:kmtksst}}

We only show that given a $k$-layered $\SST$ $\trans = (A,B,Q,\regs, q_0, \iota, \delta, \lambda, F)$, we can build an equivalent $k$-marble transducer. For this, we mimic the proof of Theorem \ref{theo:sstmt} (see \mbox{Subsection \ref{subs:proof:sstmt}}), but two difficulties arise:
\begin{itemize}
\item we used $\sum_{q,a,x} |\mar(\lambda(q,a)(x))|$ marbles to store the recursivity stack of Algorithm \ref{algo:sstmt}, we shall reduce this number to $k$;

\item we used an extra marble $\bullet$ to compute the state $q = \delta(q_0,
w[1{:}m])$ in Algorithm \ref{algo:sstmt}, in fact we can compute $q$ without
using any marble.
\end{itemize}

\subparagraph*{Removing the extra marble $\bullet$.} We first deal with this
second issue.  The problem is actually the following: given a two-way transducer
whose head is in some position $1 \le m \le |w|$ of an input $\lmark w \rmark$,
can it compute $q = \delta(q_0, w[1{:}(m-1)])$ by moving on the prefix $\lmark
w[1{:}m]$, and finally come back to position $m$?  The answer is yes: there
exists a tricky way to perform such a computation with a finite memory and
without marbles.  We shall not give the construction here, since it is well
known in the literature under the name of "lookaround removal for two-way
transducers", see e.g. \cite{chytil1977serial}.  More recently in
\cite{dartois2017reversible}, it is shown how to perform this construction by
adding only a polynomial number of states, and in $\PTIME$.

\subparagraph*{Using no marbles for $k=0$.} We first suppose that $k=0$, that is
we have a copyless $\SST$ and have to avoid using marbles, that is build a two-way
transducer.  In that case, our procedure is similar to that of
\cite{dartois2016aperiodic,dartois2017reversible}.  More precisely, we build the
same transducer as in the proof of Theorem \ref{theo:sstmt}, except that it does
not drop a marble before doing a recursive call: we execute the recursive
algorithm without recursivity stack!  To compute the content of $x$ at $m$, the
machine performs the (hardcoded) "for" loop reading $\alpha$, two cases occur:
\begin{itemize}
\item $\alpha[i] \in B$, the two-way transducer outputs it;

\item $\alpha[i] \in \regs$ moves $\lmove$ \emph{without dropping a marble}.
Using a recursive procedure, it outputs the value of $\alpha[i]$ at $m-1$.
Meanwhile, it maintains in its finite memory the current register it is working
on (this information can be updated), hence it finally knows that $\alpha[i]$
was just output.  Then it moves $\rmove$ and since the $\SST$ is copyless,
$\alpha[i] \in \regs$ occurs a most once in the whole set $\{\lambda(q,w[m])(x)
\mid x \in \regs\}$.  Therefore the machine can recover that it was computing
index $i$ of $x \mapsto \alpha$ and pursue the loop.
\end{itemize}

\subparagraph*{Using $k$ marbles with $k > 0$.} We no longer assume that $k=0$.
Intuitively, for computing recursively the content of a register $x \in \regs_k$
from a substitution $x \mapsto \alpha$, a marble transducer behaves as in the
former construction for registers of $\alpha$ that belong to layer $k$ (since
the $\SST$ is "copyless" wrt.\ registers of a same level), and we only need
to drop marks to compute the contents of registers from layers $\ell <k$.
Only $k$ marbles are needed.

\begin{lemma} 
 $\forall 0 \le j \le k$, we can build a $j$-marble transducer $\trans_j$, that has (among others) states labelled by $\bigcup_{\ell \le j} \regs_{\ell}$, such that the following is true. Let $x \in \bigcup_{\ell \le j} \regs_\ell$, $w \in A^*$ and $1 \le m \le |w|$. Then the run $(x,m,\vide) \rightarrow^* (y,m+1, \vide)$ of $\trans_j$ on $\lmark w \rmark$, ending at the first visit of position $m+1$, always exists and the output produced along this run is $\trans^{w[1:m]}(x)$.
 
\end{lemma}

\begin{proof} The proof is by induction. The base case is similar to the induction case. Assume now that $\trans_{j-1}$ is built, we build $\trans_{j}$ by adding states $x$ for $x \in \regs_j$, plus extra information collected in a finite memory (which corresponds to more states).

When in state $x \in \regs_{\ell}$ for $\ell < j$, $\trans_j$ behaves like $\trans_{j-1}$. When in state $x \in \regs_j$ and position $m$, $\trans_j$ first computes $q$ and $\alpha$ of Algorithm \ref{algo:sstmt}. Then it performs the "for" loop reading $\alpha$, and three cases occur depending on $\alpha[i]$:
\begin{itemize}

\item $\alpha[i] \in B$, $\trans_j$ outputs it;

\item $\alpha[i] \in \regs_j$, $\trans_j$ moves left to state $\alpha[i]$
\emph{without dropping a marble}.  Using a recursive procedure, it outputs the
value of $\alpha[i]$ at $m-1$.  Meanwhile, it maintains in its finite memory the
current register it is working on (this information can be updated), hence
$\trans_j$ knows that $\alpha[i]$ was just output.  Then it moves right and it
sees no marble, meaning that it did not switch to a lower layer in this
position.  Since the $\SST$ is $k$-layered, register $\alpha[i] \in \regs_j$
appears a most once in the whole set $\{\lambda(q,w[m])(x) \mid x \in
\regs_j\}$.  Therefore $\trans_j$ can remember that it was computing index $i$
of $x \mapsto \alpha$ and pursue the loop;

\item $\alpha[i] \in \regs_{\ell}$ for $\ell < j$, $\trans_j$ \emph{drops a
marble colored $j$ in the current position}, and stores in its finite memory
that it was working on index $i$ of $x \mapsto \alpha$.  Then it moves left to
state $\alpha[i]$ and executes $\trans_{j-1}$ on $\lmark w[1{:}m-1]$ to output
(by induction hypothesis) $\trans^{w[1{:}m-1]}(\alpha[i])$.  Once this is done,
$\trans_j$ moves right and meets marble $j$, meaning that it had switched to a
lower layer in this position.  Therefore it can recover from its state that it
was computing index $i$ of $x \mapsto \alpha$ and pursue the loop.

\end{itemize}

The definition of a $k$-layered $\SST$ ensures that our marble transducer uses 
not more than $j$ marbles, and only a finite auxiliary memory.
\end{proof}

To produce the output of the $\SST$, we first move to $\rmark$ and begin our backward computation.

\subparagraph*{Complexity of the construction.} The marble transducer can be constructed in $\PTIME$ from $\trans$. Indeed, when executing the recursive algorithm, we only need to store $q \in Q$, $\alpha \in \subst{\regs}{B}$ and the position $1 \le i \le |\alpha|$ in the state of the $k$-marble transducer. Hence the set of states and transitions can clearly be described in $\PTIME$. We already noted above that the extra states used for lookaround removal can also be described in $\PTIME$.

\clearpage
\section{Proof of Lemma \ref{lem:graph}}

Our objective is to describe the asymptotic growth of functions computed by
$\mb{N}$-automata.  The constructions below are very similar to those used in
\cite{weber1991degree} for computing the degree of ambiguity of
non-deterministic finite state automata.  The first step is to understand which
patterns make a function unbounded.

\begin{proposition}[\cite{sakarovitch2008decidability}] \label{prop:simon} 
  Let $\auto$ be a trim $\mb{N}$-automaton $(A,\staa, \inia,\weia,\fina)$
  computing $g \colon A^* \rightarrow \mb{N}$.  Then $g$ is bounded ($g(w) =
  \mc{O}(1)$) if and only if $\auto$ does not contain the following patterns:
\begin{itemize}

\item  a \emph{heavy cycle} on a state $q$: $\exists v \in A^+$ such that $\weia(v)(q,q) \ge 2$;

\item  a \emph{barbell} from $q$ to $q'\neq q$: $\exists v \in A^{+}$, $\weia(v)(q,q)  \ge 1$, $\weia(v)(q,q')  \ge 1$ and $\weia(v)(q',q')  \ge 1$.

\end{itemize}

\end{proposition}

The shapes of heavy cycles and barbells are depicted in Figure \ref{fig:patterns}.

\begin{figure}[h!]
    \centering
    \begin{subfigure}[b]{0.48\textwidth}
    
    	\centering
        \begin{tikzpicture}[scale=1]
       	 	\node (I) at (-1.5,0) {};
        		\node (dep) at (0,0) {$q$};
		\node (F) at (1.5,0) {};
		\draw[->, dashed] (I) edge[]  node [above] {} (dep);
		\draw[<-, dashed] (F) edge[]  node [above] {} (dep);
        		\draw[->] (dep) edge [loop above] node {\footnotesize $v, n \ge 2$} (dep);
       \end{tikzpicture}
       
        \caption{Heavy cycle on $q$}
    \end{subfigure}
    ~
    \begin{subfigure}[b]{0.48\textwidth}
    	\centering
        \begin{tikzpicture}[scale=1]
        		\node (dep) at (0,0) {$q$};
		\node (ar) at (2,0) {$q'$};
		\node (I) at (-1.5,0) {};
		\node (F) at (3.5,0) {};
		\draw[->, dashed] (I) edge[]  node [above] {} (dep);
		\draw[<-, dashed] (F) edge[]  node [above] {} (ar);
        		\draw[->] (dep) edge [loop above] node {\footnotesize $v, n_1 \ge 1$} (dep);
		\draw[->] (dep) edge[]  node [above] {\footnotesize $v, n_2 \ge 1$} (ar);
		\draw[->] (ar) edge [loop above] node {\footnotesize $v, n_3 \ge 1$} (ar);
       \end{tikzpicture}
       
        \caption{Barbell from $q$ to $q' \neq q$}
    \end{subfigure}
    \caption{\label{fig:patterns} Patterns that create unboundedness in a trim $\mb{N}$-automaton.}
\end{figure}
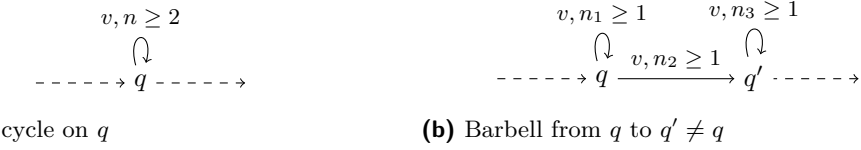

We first note that heavy cycles lead to exponential behaviors.

\begin{lemma} \label{lem:heavy} 
  Let $\auto = (A, \staa, \inia, \weia, \fina)$ be a trim $\mb{N}$-automaton
  with heavy cycles, that computes a function $g\colon A^* \rightarrow \mb{N}$.  Then
  $g$ has exponential growth.
\end{lemma}

\begin{proof} 
  The upper bound is always true.  For the lower bound, let $q \in \staa$ and $
  v \in A^+$ such that $\weia(v)(q,q) \ge 2$.  Since $\auto$ is trim $\exists u,
  w \in A^*$ such that $(\inia\weia(u))(q) \ge 1$ and $(\weia(w)\fina)(q) \ge
  1$.  Therefore $g(uv^{\ell} w) \ge 2^{\ell}$ and the result follows.
\end{proof}

% We shall see below that 
The converse of Lemma \ref{lem:heavy} also holds and thus $g$ has an exponential
growth if and only its automaton contains heavy cycles.  Furthermore, this
property is decidable in $\PTIME$.  Indeed, there is a heavy cycle on a given $q
\in \staa$ if and only if the strongly connected component of $q$ contains a
weight $\weia(a)(q_1,q_2) \ge 2$, or is unambiguous (when seen as a finite
automaton with weights $0$ or $1$).

In the sequel, let us fix a $\mb{N}$-automaton $\auto := (A,\staa,\inia,\weia,
\fina)$ without heavy cycles, that computes a function $g$.  We now show that
$g$ has $k$-polynomial growth for some $k \ge 0$.  The idea is to group states
between which there are no barbells, since they describe bounded sub-automata.
For this, we first define a graph $\graph$ that describes the barbells we can
meet.

\begin{definition}

The oriented graph $\graph$ consists in:
\begin{itemize}
\item the set of vertices $\staa$;

\item an edge $(q_1,q_2)$ if and only if there exists $q, q' \in \staa$ such that:
\begin{itemize}

\item $\exists w,w' \in A^*$ such that $\weia(w)(q_1,q) \ge 1$ and $\weia(w')(q',q_2)\ge 1$;

\item there is a barbell from $q$ to $q'$.

\end{itemize}

\end{itemize}

\end{definition}

\begin{remark} \label{rem:complexbarb} 
  The presence of a barbell from $q$ to $q'$ can be checked in $\PTIME$.
  Indeed, the set of $v \in A^+$ such that $\weia(v)(q,q) \ge 1$ is a regular
  language, for which an automaton can be constructed immediately from $\auto$.
  The same holds for $\weia(v)(q,q') \ge 1$ and $\weia(v)(q',q') \ge 1$.  We
  finally check the emptiness of their intersection.
\end{remark}

A \emph{path} in $\graph$ of length $p \ge 1$ is a sequence $(q_0,q_1)(q_1,q_2)
\dots (q_{p-1},q_p)$ of edges, and a \emph{cycle} is a path where $q_0 = q_p$.
A \emph{directed acyclic graph (dag)} is a graph without cycles.  We now show
that $\graph$ is a dag, and thus cannot have arbitrarily long paths.

\begin{lemma} \label{lem:dag} $\graph$ is a dag.
\end{lemma}

\begin{proof} 
  Assume there exists a cycle $(q_0,q_1)(q_1,q_2) \dots (q_{p-1},q_0)$.  By
  transitivity, $\exists u \in A^*$ such that $\weia(u)(q_0,q_{p-1}) \ge 1$.
  Furthermore $\exists q,q' \in \staa, w,w' \in A^*$ such that
  $\weia(w)(q_{p-1},q) \ge 1$ and $\weia(w')(q',q_0) \ge 1$, and there is a
  barbell from $q$ to $q'$.  Hence $\exists v \in A^*$ such that $\weia(v)(q,q)
  \ge 1$, $\weia(v)(q,q') \ge 1$ and $\weia(v)(q',q') \ge 1$. In particular
  $\weia(vv)(q,q') \ge 2$.  Putting everything together, $\weia(uwvvw')(q_0,q_0)
  \ge 2$, which forms a heavy cycle, a contradiction.
\end{proof}

% Since $\graph$ is a dag, it induces a partial ordering on $\staa$.  
A state is said to be \emph{minimal} if it is has no incoming edge in the dag
$\graph$.
% minimal in the ordering (equivalently, there is no state below itself in $\graph$).  
Given a state $q \in \staa$, its
\emph{height} is defined as the maximal length of a path going from a minimal
state to $q$.  We denote by $k$ the maximal height over all states, we shall see later that it
is the smallest degree of a polynomial bounding $g$, the function computed by
$\auto$.

\begin{lemma} \label{lem:lpb} There exists an infinite set of words $L$ such that $g(w) = \Omega(|w|^k)$ when $w \in L$.
\end{lemma}

\begin{proof} 
  By definition, $k$ is also the maximal length of a path, hence there exists a
  path $\pi$ of length $k$.  Therefore, we can find a sequence 
  $q_1,q'_1,\ldots,q_k,q'_k$ of $2k$ states 
%   Due to transitivity, we can assume that $\pi$ has shape $(q_1,q'_1),(q'_1,
%   q_2),\dots, (q_k,q'_k)$
  with a barbell between each $q_i,q'_i$ and a path between each $q'_i,q_{i+1}$.
  More precisely:
  \begin{itemize}

    \item $\forall 1 \le i \le k$, $\exists v_i \in A^+$ such that $\weia(v_i)(q_i,q_i)\ge 1$ and $\weia(v_i)(q_i,q'_i)\ge 1$ and $\weia(v_i)(q'_i,q'_i)\ge 1$;

    \item $\forall 1 \le i < k$,  $\exists u_i \in A^*$ such that $\weia(u_i)(q'_i,q_{i+1})\ge 1$.

  \end{itemize}
  Now consider the words $w_{\ell} := v_1^{\ell} u_1 v_2^{\ell} \cdots u_{k-1}
  v_k^{\ell}$ for ${\ell} \ge 0$ (the intuition is that $w_{\ell}$ "loops"
  ${\ell}$ times in each barbell), then $\weia(w_{\ell})(q_1,q'_k) \ge
  {\ell}^k$.  Since $\auto$ is trim, there exists $u,v \in A^*$ such that
  $g(uw_{\ell}v) \ge {\ell}^k = \Omega({|uw_{\ell}v|}^k)$ when ${\ell}
  \rightarrow + \infty$.
\end{proof}

Finally, we consider the partition $S_0, \dots, S_k$ of $\staa$, where $S_i := \{\text{states of height } i\}$. We now show that it verifies the properties of the last point of Lemma \ref{lem:graph}.

\begin{lemma} \label{lem:prop} The following statements hold:
\item

\begin{enumerate}

\item \label{po:growth} $\forall q \in S_i$, if $\exists w \in A^*$ such that $\weia(w)(q,q') \ge 1$, then $q' \in S_j$ for some $i \le j$.

\item \label{po:prop:2} $\exists B \ge 0$ such that $\forall 0 \le i \le k$, $\forall q,q' \in S_i, \forall w \in A^*$, $\weia(w)(q,q')\le B$;

\item $\forall 0 \le i \le k$, $\exists B_i , C_i\ge 0$ such that $\forall q,q' \in \bigcup_{j \le i} S_j, w \in A^*$, $\weia(w)(q,q') \le B_i |w|^i + C_i$.

\end{enumerate}

\end{lemma}

\begin{proof} \begin{enumerate}

\item Suppose that $\weia(w)(q,q') \ge 1$.  Then, every path in $\graph$ from a
minimal state $m$ to $q$ of length $p$ can be completed in a path from $m$ to
$q'$ (of length at least $p$).  Thus the height of $q$ is at most the height of
$q'$.

\item Suppose there is a barbell from $q$ to $q'$.  Then, every path in $\graph$
from a minimal state $m$ to $q$, of size $p$, can be extended to a path from
$m$ to $q'$, of size $p+1$.  Thus the height of $q'$ is strictly more than that
of $q$.  Hence, there are no barbells inside each $S_i$.  Using Proposition
\ref{prop:simon}, the sub-automaton "induced" on $S_i$ is necessarily bounded
(we also need point \ref{po:growth} above to show that runs from $q$ to $q'$
necessarily stay all the time in $S_i$).

\item We use matrices to shorten the notations. Given a matrix $A$, we denote by $\sup A$ the maximum of its coefficients. The result is shown by induction on $0 \le i \le k$. 

Let $U = \bigcup_{j \le i-1} S_i$.  We denote by $\mu_U\colon A^* \rightarrow
\mb{N}^{U \times U}$ the co-restriction of $\mu$ to $U$, that is the morphism $w
\mapsto (\mu(w)(q,q'))_{q,q' \in U}$.  Let $S := S_{i}$, we define $\mu_S \colon w
\mapsto (\mu(w)(q,q'))_{q,q' \in S}$ in a similar way.  Finally, we set
$\mu_{U,S}\colon A^* \rightarrow \mb{N}^{U \times S}, w \mapsto (\mu(w)(q,q'))_{q \in
U, q' \in S}$.

The induction hypothesis reformulates as $\sup \mu_{U}(w) \le B_{i-1} |w|^{i-1}
+ C_{i-1}$.  It follows from point \ref{po:growth} above that $\forall w\in A^*$:
$$
\mu(w) = \begin{pmatrix} 
\mu_U(w) & \mu_{U,S}(w) \\
0 & \mu_S(w)
\end{pmatrix}
$$
By matrix multiplication we also have for all $w \in A^*$:
$$
\mu(w) = \begin{pmatrix} 
\mu_U(w) & 
\displaystyle \sum_{1 \le m \le |w|} \mu_{U}(w[1{:}(m-1)])\mu_{U,S}(w[m]) \mu_{S}(w[(m+1){:}|w|])
\\ 0 & \mu_S(w)
\end{pmatrix}
$$
Recall that $\sup\mu_{U}(w) \le B_{i-1} |w|^{i-1} + C_{i-1}$; we also have $\sup\mu_S(w)
\le B$ by point \ref{po:prop:2}.  Hence to show the property on $\sup \mu(w)$,
we only need to consider the submatrix
$$
\mu_{U,S}(w) = \sum_{1 \le m \le |w|} \mu_{U}(w[1{:}(m-1)])\mu_{U,S}(w[m]) 
\mu_{S}(w[(m+1){:}|w|]) \,.
$$
But $\sup \mu_{U}(w[1{:}(m-1)]) \le B_{i-1} (m-1)^{i-1} + C_{i-1}$.  We also have
that $\sup \mu_{U,S}(w[m])$ and $\sup \mu_{S}(w[(m+1){:}|w|])$ are bounded.
Hence the $\sup$ of each matrix in the sum is bounded by $B' |w|^{i-1} + C'$ for
some $B', C' \ge 0$.  Finally
$$
\sup \mu_{U,S} \le |w| (B' |w|^{i-1} + C') \le B_i |w|^i + C_i \,.
$$
for well-chosen $B_i, C_i \ge 0$.
\qedhere
\end{enumerate}
\end{proof}

Finally $g(w) = \mc{O}(|w|^k)$, and it is reached asymptotically on $L$ (by
Lemma \ref{lem:lpb}).  Hence $g$ has $k$-polynomial growth.  Note that $S_0,
\dots, S_k$ can be computed in $\PTIME$.  How?  First, we compute the edges of
$\graph$, by detecting the presence of a barbells between two states in $\PTIME$
(see above).  Once $\graph$ is built, the height is computed by a
graph traversal in linear time.

\section{Proof of Lemma \ref{lem:boundedcopy}}

Let $\trans := (A,B, \regs, \iota, \lambda, F)$ be the simple $\SST$ computing a
function with $(k+1)$-polynomial growth.  Let $S_0, \dots, S_{k+1}$ be the partition
given by Lemma \ref{lem:graph} applied to
$\fauto{\trans}=(A,\regs,\inia,\weia,\fina)$, together with the constant $B$.
The idea is to remove the registers of $S_0$ since they contains only bounded
strings.  Formally, from Lemma \ref{lem:graph} we deduce that there exists $L
\ge 0$ such that $|\trans^w(x)| \le L$ for all $x \in S_0$ and $w \in A^*$.  So
the strings $\trans^w(x)$ can be hardcoded in the states of the machine.

We define the $\SST$ $\srans := (A,B,Q,\regp, q_0, \iota', \delta, \lambda', G)$ as follows:
\begin{itemize}
\item the set $Q$ is $S_0 \rightarrow A^{\leq L}$. It represents the possible valuations of the registers from $S_0$;

\item the initial state $q_0$ is given by $q_0(x)=\iota(x)$ for $x\in S_0$, which initializes the registers to their initial values;

\item the state $\delta(q,a)$ is the function $x \mapsto q(\lambda(a)(x))$ where $q$ is seen as a substitution in $\subst{S_0}{B}$, i.e. it replaces each register by its value given by the state. This definition makes sense since only registers from $S_0$ can occur in $\lambda(a)(x)$ (because $x \in S_0$);

\item the register set $\regp$ is $ \biguplus_{1 \le i \le k+1} S_i$;

\item the initial function $\iota'$ maps $y \mapsto \iota(y)$ for $y\in\regp$;

\item the update function is such that $\lambda'(q,a)(y) = q(\lambda(a)(y))$;

\item all states are final and $G(q) := q(F)$.
\end{itemize}

It is easy to see that $\delta(q_0,w)(x) = \trans^{w}(x)$ for all $x \in S_0$.
Therefore, $\srans$ and $\trans$ are equivalent.  Moreover, by definition of
$\fauto{\trans}$, for all registers $x,y\in\regs$ and all words $w\in A^{*}$,
the value $\mu(w)(y,x)$ is the number of occurrences of $y$ in $\lambda(w)(x)$.
Using Lemma~\ref{lem:graph}, we deduce that $\srans$ is $(k,B)$-bounded with the
partition $S_1, \dots, S_{k+1}$ of $\regp$.

\section{Proof of Lemma \ref{lem:kbounded}}

It is known from \cite{alur2012regular,dartois2016aperiodic} that a $(0,B)$-bounded $\SST$ can be converted in a copyless $\SST$. Our objective is to generalize their proof to move from $(k,B)$-bounded $\SST$  to $k$-layered.

Recall that each layer of a $(k,B)$-bounded (resp.  $k$-layered) $\SST$ is
somehow a $(0,B)$-bounded (resp.  copyless) $\SST$, that can also "call"
registers from the lower layers.  So, the main idea is to convert each bounded copy
layer into a locally copyless layer.  The difficulty is to take into account
copies coming from the lower layers.  This is done by using the intermediate
model of \emph{$\SST$ with external functions} ($\SSTF$), defined below.  It
corresponds to an $\SST$ with a set of functions $\oras$ that can be called in
an oracle-like style.  The output of these functions can then be used in the
substitutions along a run.

We then discuss the properties of these $\SSTF$.  This way, we shall be able to
perform an induction on a $(k,B)$-bounded $\SST$, by making the layers copyless
one after the other.

\begin{definition} An \emph{$\SST$ with external functions} ($\SSTF$) $\trans = (A,B,Q, \regs, \oras, q_0, \iota, \delta,\lambda, F)$ consists of:

\item

\begin{itemize}

\item an input alphabet $A$ and an output alphabet $B$;

\item a finite set of states $Q$ with an initial state $q_0 \in Q$;

\item a finite set $\regs$ of registers;

\item a finite set $\oras$ of total functions from $A^* \rightarrow B^*$;

\item an initial function $\iota \colon \regs \rightarrow B^*$;

\item a (partial) transition function $\delta \colon Q \times A \rightarrow Q$;

\item a (partial) register update function $\lambda\colon Q \times A \rightarrow \subst{\regs}{B \cup \oras}$ with same domain as $\delta$;

\item a (partial) output function  $F\colon Q \rightarrow (\regs \cup B)^*$.

\end{itemize}
\end{definition}

The machine $\trans$ defines a (partial) function $f \colon A^* \rightarrow B^*$ as follows. Let us fix $w \in A^*$. If there is no accepting run of the \emph{one-way} automaton $(A, Q,q_0, \delta, \dom(F))$ over $w$, then $f(w)$ is undefined. Otherwise, let $q_m:= \delta(q_0,w[1{:}m])$ be the $m$-th state of this run. We define for $0 \le m \le |w|$, $\trans^{w[1:m]} \colon \regs \rightarrow B^*$ ("the values of the registers after reading $w[1{:}m]$") as follows:
\begin{itemize}

\item $\trans^{w[1:0]} (x) = \iota(x)$ for all $x \in \regs$;

\item for $1 \le m \le |w|$, we define $\lambda_{w[1:m]} \in \subst{\regs}{B}$
to be the substitution $\lambda(q_{m-1}, w[m])$ in which the external function names
$\exte$ are replaced by the values $\exte(w[1{:}m])$.  Formally, let $r \in
\subst{\oras}{\regs \cup B}$ which maps $\exte \in \oras \mapsto
\exte(w[1{:}m])$, then $\lambda_{w[1:m]} \colon x \mapsto r(\lambda(q_{m-1}, w[m])(x))$.

Then we let $\trans^{w[1:m]} := \trans^{w[1:(m-1)]}\circ \lambda_{w[1:m]}$.
\end{itemize}
\begin{example} Assume that $\regs = \{x\}$,  $\oras = \{{\exte}\}$, $\trans^{w[1{:}(m-1)]}(x) = ab$, $\exte(w[1{:}m]) = cc$ and $\lambda(q_{m-1}, w[m])(x) = x\exte b$. Then $\trans^{w[1{:}m]}(x) = abccb$.
\end{example}

Finally we set $f(w) := \trans^{w}(F(q_{|w|})) \in B^*$.

\begin{remark}
  An $\SSTF$ such that $\oras = \varnothing$ is just an $\SST$, and the
  semantics coincide.
\end{remark}

An $\SSTF$ is said to be \emph{copyless} whenever it does not duplicate
\emph{its registers}.  Formally, it means that $\forall x \in \regs, q \in Q, a
\in A$, $x$ occurs at most once in the set $\{\lambda(q,a)(y)\mid y \in
\regs\}$.  There are no restrictions on the use of external functions $\exte \in
\oras$, since they intuitively correspond to "lower layers" of a $k$-layered
$\SST$.

An $\SSTF$ is said to be $B$-\emph{bounded} if it is $(0,B)$-bounded in the
sense of Definition \ref{def:bounded}.  Formally, we define $\lambda_{u}^{v}$ to
be the substitution applied when reading $v=a_1\cdots a_\ell \in A^*$ after having read $u \in
A^*$, that is $\lambda_{ua_1} \circ \lambda_{ua_1a_2} \circ \cdots \circ \lambda_{uv}$.
Then the machine is $B$-bounded if $\forall u,v \in A^*$ and $x \in \regs$, $x$
occurs at most $B$ times in $\{\lambda_{u}^{v}(y) \mid y \in \regs\}$.

\begin{lemma} \label{lem:mainSSTF} Given a $B$-bounded $\SSTF$, one can build an equivalent copyless $\SSTF$ that uses the same external functions.
\end{lemma}

\begin{remark} 
  An $\SSTF$ may have no finite presentation, since the functions from $\oras$
  are not required to be computable (even if in practice it will not be the
  case).  In order to have an effective Lemma \ref{lem:mainSSTF}, we consider
  that an $\SSTF$ only contains the \emph{function names} in its representation
  (the functions themselves being given like some oracles).
\end{remark}

The subsections \ref{sec:boundedcopies-unambiguity} and
\ref{sec:remove-ambiguity} below are dedicated to the proof of
Lemma~\ref{lem:mainSSTF}.  To simplify the matters, we shall only reason about
total transducers, but the result is the same with domains.

Assume Lemma \ref{lem:mainSSTF} holds, we  now prove Lemma \ref{lem:kbounded}, that is given a  $(k,B)$-bounded $\SST$, one can build an equivalent $k$-layered $\SST$.

\begin{proof}[Proof of Lemma \ref{lem:kbounded}.]
  The proof is done by induction on $k \ge 0$. 
%   and the induction case goes over
%   the base case, therefore we do not deal with it.  
  Let $\trans$ be a $(k,B)$-bounded $\SST$ computing a function $f$. By 
  induction we assume that $(k',B')$-bounded $\SSTs$ with $k'<k$ can be 
  converted to $k'$-layered $\SSTs$.
  \begin{enumerate}

    \item We build a bounded $\SSTF$ $\srans$ that computes $f$, whose external
    functions are computable by $(k-1)$-layered $\SSTs$.  Let $\regs_0, \dots,
    \regs_{k}$ be the layers of registers of $\trans$, and let $\regu :=
    \bigcup_{0 \le i < k} \regs_i$.  For all $x \in \regu$ (empty if $k=0$),
    define the function ${\exte}_x\colon w \mapsto \trans^w(x)$ that describes "the
    value of $x$ after reading $w$"; it is computed by a $(k-1,B)$-bounded $\SST$
    derived from $\trans$.  We then transform the layer $\regs_{k}$ in a
    $B$-bounded $\SSTF$ $\srans$ whose external functions are $\{ {\exte}_x \mid
    x \in\regu\}$.

    \item By Lemma \ref{lem:mainSSTF}, we can transform $\srans$ in a copyless
    $\SSTF$ $\srans'$ that also uses the $\{ {\exte}_x \mid x \in \regu \}$.

    \item By induction hypothesis, the functions ${\exte}_x$ are computable by
    $(k-1)$-layered $\SSTs$.

    \item Finally we build the $k$-layered $\SST$ for $f$.  Using a product
    construction, we compute "in parallel" all the functions ${\exte}_x$ for
    $x\in\regu$ in a $(k-1)$-layered $\SST$, and use it as layers $0, \dots,
    k-1$ of $\srans'$.  \qedhere
\end{enumerate}
\end{proof}

\subsection{From bounded copies to 
unambiguity}\label{sec:boundedcopies-unambiguity}

In order to move from a bounded $\SSTF$ to a copyless $\SSTF$, the natural idea is to use copies of each register. However, we cannot maintain $B$ copies of each variable all the time: suppose that $x$ is used both in $y$ and $z$. If we have $B$ copies of $x$, we cannot produce in a copyless way $B$ copies of $y$ and $B$ copies of $z$.

To solve this issue, we shall follow the ideas of \cite{dartois2016aperiodic} (which have no external functions). We shall maintain $n_x$ copies of $x$ \emph{if this register is involved exactly $n_x$ times in the final output}. We will thus have enough copies to produce the output in a copyless fashion. However, this number $n_x$ cannot be computed before reading the whole input. Our transducer will have to guess it, what motivates the introduction of \emph{nondeterminism} below. In fact, we shall do better than nondeterminism and obtain an \emph{unambiguous} machine.

\subparagraph*{Nondeterministic $\SSTs$.} A \emph{non-deterministic $\SST$ with
external functions} ($\NSSTF$ for short) $\mc{N} = (A,B,Q, \regs, \oras, I,
\Delta, \Lambda, F)$ is defined similarity as $\SSTF$, except that the
underlying automaton is non-deterministic.  The reader can refer to
\cite{alur2012regular} for examples of non-deterministic $\SSTs$ without external
functions.  Formally, the changes are the following:
\begin{itemize}
 \item the initial state $q_0$ and initial function $\iota$ are replaced by a 
 partial initial function $I\colon Q \rightarrow (\regs \rightarrow B^*)$. A state $q$ is said to be  \emph{initial} when $q \in \dom(I)$. Intuitively, $I(q)$ describes how to initialize the registers if we start from state $q$;
 \item the transition function $\delta$ is replaced by a transition relation $\Delta \subseteq Q \times A \times Q$;
 \item the update function $\lambda$ is replaced by a mapping $\Lambda \colon \Delta \rightarrow \subst{\regs}{\oras \cup B }$ which maps every transition to a substitution.
 \end{itemize}
A run of $\utrans$ is a run of the non-deterministic automaton $\auto := (A,Q,\dom(I), \Delta,\dom(F))$; it is said \emph{initial} if it starts in an initial state, and \emph{accepting} if it also ends in a final state. The transducer describes a relation $R \subseteq A^* \times B^*$. We have $(w,v) \in R$ when $v$ is produced on some accepting run labelled by $w$ (updates along a fixed run are defined as for an $\SSTF$).

The $\NSSTF$ is said to be \emph{copyless} if $\Lambda$ maps to copyless
substitution.  It is said to be \emph{unambiguous} if $\auto$ is so, i.e. there
is at most one accepting run for each input word.  In that case, we can
consider that it computes a partial function.

\begin{lemma} Given a bounded $\SSTF$, one can build an equivalent unambiguous copyless $\NSSTF$ that uses the same external functions.
\end{lemma}

Let $\trans = (A,B,Q, \regs, \oras, q_0, \iota, \delta, \lambda, F)$ be a
bounded total $\SSTF$, we show how to build an equivalent unambiguous copyless
$\NSSTF$ $\utrans = (A,B,P, \regp, \oras, I, \Delta, \Lambda, G)$.  Before
detailing the construction, we fix some notations concerning the occurrences
of the registers.

\subparagraph*{Occurrences in the final output.} After reading some prefix $u \in A^*$, we want to compute the number of times $x \in \regs$ is used in the final output after reading the suffix $v \in A^*$. Thus we define the $N_u^v(x)$ as the number of occurrences of $x$ in $(\lambda_u^v)(F_{\delta(q_0,uv)})$.

\begin{remark} 
  Contrary to $\lambda_u^v(x)$ which may refer to external functions, $N_u^v(x)$
  does not depend on external functions.  Indeed, they do not change how the
  registers flow in each other during the substitutions.  Also, the only
  useful information about $u$ is $\delta(q_0,u)$.
\end{remark}

Since the $\SSTF$ is bounded, the $N_{u}^v(x)$ are bounded by some $B \ge 0$. We can even be more precise and note that they are somehow preserved along a run.

\begin{example} \label{ex:pos} 
  Assume that $\regs = \{x,y\}$ and $\oras = \{{\exte}\}$.  Let $w = uav \in
  A^*$ with $a \in A$.  Assume that the substitution applied by $\trans$ in
  state $\delta(q_0,u)$ when reading $a$ is $ x \mapsto x, y \mapsto xy{\exte}$.
  If $N_{u}^{av}(x) = n$, then $N_{ua}^{v}(x) + N_{ua}^{v}(y) = n$.  Indeed, the
  $n$ occurrences of $x$ in the final output are "transformed" in occurrences of
  either $x$ xor $y$ after reading $a$.
\end{example}

The previous example can easily be generalized to obtain Claim \ref{claim:eq}.

\begin{claim} \label{claim:eq} Forall $u,v \in A^*, a \in A, x \in \regs$, we have:

$$N_u^{av}(x) = \sum_{y \in  \regs} c^y_x \times N_{ua}^v (y)$$

where $c_x^y$ is the number of occurrences of $x$ in $\lambda(\delta(q_0,u),a)(y)$.
\end{claim}

\subparagraph*{States and registers of $\utrans$. } 
The states of $\utrans$ are $P := Q \times (\regs \rightarrow \{0, \dots, B\})$.
The registers are $\regp := \regs \times \{1, \dots, B\}$.  Consider the $m$-th
configuration of the unique accepting run of $\utrans$ labelled by $w \in A^*$.
Let $u = w[1{:}m]$ and $v = w[(m+1){:}|w|]$.  We want to keep track of:
\begin{enumerate}
\item \label{inv:1} in the first component of $P$: the state $\delta(q_0,u)$ of $\trans$;
\item \label{inv:2} in second component of $P$: the function $N_{u}^{v} := x \mapsto N_{u}^{v} (x)$;
\item \label{inv:3} in the registers $(x,1), \dots, (x,{N_{u}^{v} (x)})$: the same word $\trans^u(x)$ (it corresponds to the copies);
\item \label{inv:4} in the registers $(x,{N_{u}^{v} (x)+1}), \dots, (x,{B})$: the word $\vide$ (they are not used).
\end{enumerate}
We refer to these points as the \emph{invariants} maintained along the unique accepting run of $\utrans$.

\subparagraph*{Initial states of $\utrans$.} A state of $p = (q,g)$ is initial if and only $q = q_0$. Indeed, we start from the initial state of $\trans$ on the first component. For the second component, we shall use nondeterminism to guess the $N_{\vide}^w(x)$ in the beginning of a run. We initialize the registers according to Invariants \ref{inv:3} and \ref{inv:4}, that is $I(p)(x,1) = \cdots = I(p)(x,{g(x)}) := \iota(x)$ and $I(p)(x,{g(x) + 1}) =  \cdots = I(p)(x,{B}) := \vide$.

\subparagraph*{Final states and output of $\utrans$.} A state $p = (q,g)$ is final if and only if:
\begin{itemize}
\item $q \in \dom(F)$, i.e. the computation of $\trans$ is accepting;
\item $\forall x \in \regs$, $g(x)$ is the number of occurrences of $x$ in $F(q)$, i.e. $g$ describes correctly the occurrences of the registers in $F(q)$.
\end{itemize}
The output $G(q,g)$ is defined as $\mu_g(F(q))$ where $\mu_g$ replaces the $i$-th occurence of $x \in F(q)$ by $(x,i) \in \regp$. By definition of final states, these registers exist.

\subparagraph*{Updates of $\utrans$.} 
We want to define $\Delta$ in order to maintain Invariants \ref{inv:1} -
\ref{inv:4} along the accepting runs of $\utrans$.  Intuitively, $\utrans$
behaves as $\trans$ on the first component of the state (without using
nondeterminism), and guesses the correct number of copies in the second
component.

Formally, we have $((q,g),a,(q',g')) \in \Delta$ if and only if:
\begin{itemize}
\item $q' = \delta(q,a)$;
\item $\forall x \in \regs, g(x) =  \sum_{y \in  \regs} c^y_x \times g' (y)$ where $c_x^y$ is the number of occurrences of $x$ in $\lambda(q,a)(y)$. This way we make a guess that respects Claim \ref{claim:eq}.
\end{itemize}

\begin{example} \label{ex:guess}  Assume that $\regs = \{x,y\}$, $\oras = \{{\exte}\}$, $\delta(q,a)=q'$ and $\lambda(q,a)=x \mapsto x, y \mapsto xy{\exte}$ in $\trans$. Let $p := (q, x \mapsto 2, y \mapsto 1)$, $p' := (q', x \mapsto 1, y \mapsto 1)$ and $p'' := (q', x \mapsto 2, y \mapsto 0)$  be states of $\utrans$. We have  $(p,a,p') \in \Delta$ but $(p,a,p'') \not \in \Delta$.
\end{example}

The function  $\Lambda$ is defined accordingly to the guesses we do in $\Delta$, in order to maintain Invariants \ref{inv:3} and \ref{inv:4} along the accepting run. This can be performed in a copyless fashion with respect to $\regp$, since by respecting Claim \ref{claim:eq} we do not create more copies than existing before. The external functions are used as they were in $\trans$.

\begin{example} Following Example \ref{ex:guess}, we assume that $ \regp =: \{(x,1), (x,2), (y,1), (y,2)\}$. Then  $\Lambda(p,a,p') = (x,1) \mapsto (x,1), (x,2) \mapsto \vide, (y,1) \mapsto (x,2) (y,1) \exte, (y,2) \mapsto \vide$.
\end{example}

\subparagraph*{Properties of $\utrans$.} We already noted that $\utrans$ is copyless, it remains to show  that it is unambiguous and describes the  same function as $\trans$. Intuitively, if a wrong guess is done in a run, it will propagate until the end and lead to a non-accepting state.

\begin{claim} For $w \in A^*$, there is a unique accepting run $\rho = p_0 \rightarrow \cdots \rightarrow p_{|w|}$ labelled by $w$. More precisely, for all $0 \le m \le |w|$, we have $p_m = \left(\delta(q_0,w[1{:}m]), N_{w[1:m]}^{w[(m+1):|w|]}\right)$.
\end{claim}

\begin{proof} The run $\rho$ is clearly accepting. Let $\rho' =  p'_0 \rightarrow \cdots \rightarrow p'_{|w|}$ be another accepting run labelled by $w$, and $0 \le m \le |w|$ be the largest index such that $p_m \neq p'_m$. Let $u = w[1{:}m]$ and $v = w[(m+1){:}|w|]$. Then $p'_m = (\delta(q_0,u), g)$ by construction of the transitions, thus $g \neq N_{u}^{v}$.

Necessarily $m < |w|$ since $p'_m$ is not final. Therefore $v = av'$ for some $a \in A$ and $p'_{m+1} = (\delta(q_0,ua),N_{ua}^{v'})$. By construction of the transitions, for all $x \in \regs$ we have $g(x) = \sum_{y \in \regs} c_{x}^y N_{ua}^{v'}(y)$ where $c_x^y$ is the number of occurrences of $x$ in $\lambda(\delta(q_0,u),a)(y)$. By Claim \ref{claim:eq} it means that $g = N_{u}^{av'}$, a contradiction.
\end{proof}

\begin{claim} Invariants \ref{inv:3} and \ref{inv:4} are preserved along the run $\rho$.
\end{claim}

It follows from the construction of the output that $\utrans$ and $\trans$ are equivalent.

\subsection{Removing unambiguity}\label{sec:remove-ambiguity}

\begin{lemma} 
  Given an unambiguous copyless $\NSSTF$, one can build an equivalent copyless
  $\SSTF$ that uses the same external functions.
\end{lemma}

Let $\utrans = (A,B,Q, \regs, \oras, I, \Delta, \Lambda, F)$ be an unambiguous
transducer.  We follow the proof of \cite{alur2012regular}, which works for
$\SSTs$ (without external functions) over infinite strings.  Due to the
similarity with an existing proof, we adopt here a more informal style.

Without loss of generality, we assume that $\utrans$ is \emph{trim}, in the
sense that every state of its underlying automaton $(A,Q,I,\Delta,F)$ is both
reachable from $I$ and co-reachable from $F$.

% paul: changed A_u to Q_u since A_u looks like an subalphabet and not subset of states
\subparagraph*{Wrong proof idea.} A first idea is to determinize the transducer
with a subset construction, as we do for finite automata.  After reading a
prefix $u \in A^*$, the $\SSTF$ will keep track of the set $Q_u \subseteq Q$ of
states that can be reached after reading $u$.  Formally, we have $p \in Q_u$
if and only if there exists an initial run labelled by $u$ that ends in $p$.

Given an initial run $\rho$, we denote by $\trans^{\rho}(x)$ "the value of $x$
after following run $\rho$", defined for deterministic $\SSTF$ by composing the
substitutions along the transitions of the run.

Since $\utrans$ is unambiguous, we deduce that for all $p \in Q_u$, there exists
\emph{exactly one initial run labelled by $u$ that ends in $p$}, that we denote
$\rho_p$.  
% Indeed, if there exists two distinct runs, then $p$ cannot be co-accessible. 
In order to keep the whole behavior of $\utrans$, we also store
$\trans^{\rho_p}(x)$ for all $x \in \regs$.  However, it is not possible to
update this information in a copyless manner, as shown in the example below.

\begin{example} 
  Assume that $Q_u = \{q\}$, the next letter is $a$, and we have $(q,a,q_1)
  \in \Delta$ and $(q,a,q_2) \in \Delta$.  Also suppose that $\regs = \{x\}$ and
  $\Lambda((q,a,q_1)) = \Lambda((q,a,q_2)) = x \mapsto x $.  Then our procedure
  is supposed to make two copies of $x$ for $\rho_{q_1}$ and $\rho_{q_2}$.
\end{example}

\subparagraph*{Storing substitutions along partial runs.} The wrong proof presented above gives a general idea for our construction. We now build a truly copyless $\SSTF$ $\srans$ that is equivalent to $\utrans$. Instead of storing the contents of the registers, $\srans$ shall keep in memory \emph{the copyless substitutions applied along partial runs of $\utrans$}. When a run branches (i.e. when a state has several successors) we create one new substitution for each new branch that appears. This way, we shall avoid the copy issue that appeared in the former wrong proof. Due to non-ambiguity, we only need to store a finite number of substitutions.

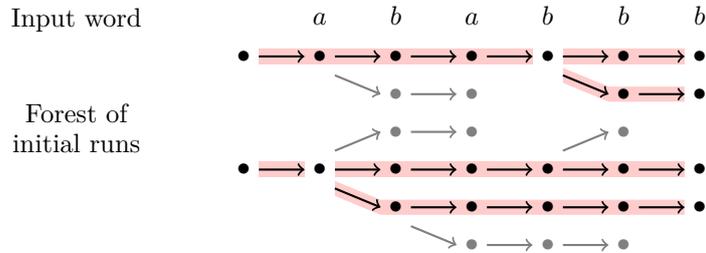
\begin{figure}[h!]

\begin{center}
\begin{tikzpicture}[scale=1]

        \fill[red!20] (0.2,-0.6) -- (0.2,-0.4) -- (3.8,-0.4) -- (3.8,-0.6) -- cycle;
        		\fill[red!20] (4.2,-0.6) -- (4.2,-0.4) -- (5.8,-0.4) -- (5.8,-0.6) -- cycle;
		\fill[red!20]  (4.2,-0.65) -- (4.2,-0.85) -- (4.8,-1.1) -- (5.8,-1.1) -- (5.8,-0.9) -- (4.8,-0.9) -- cycle;
        
        \fill[red!20] (0.2,-1.9) -- (0.2,-2.1) -- (0.8,-2.1) -- (0.8,-1.9) -- cycle;
        		\fill[red!20] (1.2,-1.9) -- (1.2,-2.1) -- (5.8,-2.1) -- (5.8,-1.9) -- cycle;
		\fill[red!20]  (1.2,-2.15) -- (1.2,-2.35) -- (1.8,-2.6) -- (5.8,-2.6) -- (5.8,-2.4) -- (1.8,-2.4) -- cycle;
		
	\node (in) at (-2.2,-0.3) [above]{ Input word};

	\node (in) at (-2.2,-1.5) [above]{Forest of};
	\node (in) at (-2.2,-1.9) [above]{initial runs};

	\node (in) at (0,-0.5) []{ $\bullet$};
	\draw[->,thick](0.2,-0.5) to (0.8,-0.5);	
	\node (in) at (1,-0.5) []{ $\bullet$};
	\draw[->,thick](1.2,-0.5) to (1.8,-0.5);	
	\node (in) at (2,-0.5) []{ $\bullet$};
	\draw[->,thick](2.2,-0.5) to (2.8,-0.5);		
	\node (in) at (3,-0.5) []{ $\bullet$};
	\draw[->,thick](3.2,-0.5) to (3.8,-0.5);		
	\node (in) at (4,-0.5) []{ $\bullet$};
	\draw[->,thick](4.2,-0.5) to (4.8,-0.5);			
	\node (in) at (5,-0.5) []{ $\bullet$};
	\draw[->,thick](5.2,-0.5) to (5.8,-0.5);	
	\node (in) at (6,-0.5) []{ $\bullet$};
	
		\draw[->,thick,gray](1.2,-0.75) to (1.8,-1);		
		\node[gray] (in) at (2,-1) []{  $\bullet$};
		\draw[->,thick,gray](2.2,-1) to (2.8,-1);		
		\node[gray] (in) at (3,-1) []{  $\bullet$};
		
		\draw[->,thick](4.2,-0.75) to (4.8,-1);		
		\node (in) at (5,-1) []{  $\bullet$};
		\draw[->,thick](5.2,-1) to (5.8,-1);		
		\node (in) at (6,-1) []{  $\bullet$};

	\node (in) at (0,-2) []{ $\bullet$};
	\draw[->,thick](0.2,-2) to (0.8,-2);	
	\node (in) at (1,-2) []{ $\bullet$};
	\draw[->,thick](1.2,-2) to (1.8,-2);	
	\node (in) at (2,-2) []{ $\bullet$};
	\draw[->,thick](2.2,-2) to (2.8,-2);		
	\node (in) at (3,-2) []{ $\bullet$};
	\draw[->,thick](3.2,-2) to (3.8,-2);		
	\node (in) at (4,-2) []{ $\bullet$};
	\draw[->,thick](4.2,-2) to (4.8,-2);			
	\node (in) at (5,-2) []{ $\bullet$};
	\draw[->,thick](5.2,-2) to (5.8,-2);	
	\node (in) at (6,-2) []{ $\bullet$};
	
		\draw[->,thick,gray](1.2,-1.75) to (1.8,-1.5);		
		\node[gray] (in) at (2,-1.5) []{  $\bullet$};
		\draw[->,thick,gray](2.2,-1.5) to (2.8,-1.5);		
		\node[gray] (in) at (3,-1.5) []{  $\bullet$};
		
		\draw[->,thick,gray](4.2,-1.75) to (4.8,-1.5);		
		\node[gray] (in) at (5,-1.5) []{  $\bullet$};

		\draw[->,thick](1.2,-2.25) to (1.8,-2.5);		
		\node (in) at (2,-2.5) []{  $\bullet$};
		\draw[->,thick](2.2,-2.5) to (2.8,-2.5);		
		\node (in) at (3,-2.5) []{  $\bullet$};
		\draw[->,thick](3.2,-2.5) to (3.8,-2.5);		
		\node (in) at (4,-2.5) []{  $\bullet$};
		\draw[->,thick](4.2,-2.5) to (4.8,-2.5);	
		\node (in) at (5,-2.5) []{  $\bullet$};
		\draw[->,thick](5.2,-2.5) to (5.8,-2.5);		
		\node (in) at (6,-2.5) []{  $\bullet$};	
		
			\draw[->,thick,gray](2.2,-2.75) to (2.8,-3);		
			\node[gray] (in) at (3,-3) []{  $\bullet$};
			\draw[->,thick,gray](3.2,-3) to (3.8,-3);		
			\node[gray] (in) at (4,-3) []{  $\bullet$};
			\draw[->,thick,gray](4.2,-3) to (4.8,-3);		
			\node[gray] (in) at (5,-3) []{  $\bullet$};

	\node[above] (in) at (1,-0.2) []{  $a$};
	\node[above] (in) at (2,-0.2) []{  $b$};
	\node[above] (in) at (3,-0.2) []{  $a$};
	\node[above] (in) at (4,-0.2) []{  $b$};
	\node[above] (in) at (5,-0.2) []{  $b$};
	\node[above] (in) at (6,-0.2) []{  $b$};

       \end{tikzpicture}
\end{center}
       \caption{\label{fig:branches} Example of forest  on input $ababbb$.}
\end{figure}

This idea is better explained with a picture.  In Figure \ref{fig:branches}, we
have drawn the \emph{forest of initial runs} of $\utrans$ on $u = ababbb$.
This forest describes all the initial runs labelled by $u$.  States are depicted
with $\bullet$, and the transition relation by $\rightarrow$.  The leftmost
states correspond to the initial states (here $|I| = 2$), and the rightmost
describe the set $Q_u$ (here $|Q_u| = 4$).

Several subtrees of the forest are drawn in gray, we call them the \emph{dead
subtrees}.  They correspond to the partial runs that got blocked at some point
("wrong guesses")
% Even if they were a possibility of final run in the past, 
and they are no longer useful.
When removing the dead subtrees from the forest of initial runs, we obtain the \emph{alive forest} drawn in black. The \emph{deterministic alive branches} of the alive forest are overlined in red. They correspond to the longest portions of runs that do not split in several branches. In our example, there are $6$ such deterministic branches. Due to unambiguity, the alive forest has only a bounded number of deterministic alive branches.

Thus $\srans$ will keep track of the following information:
\begin{itemize}
\item the substitutions applied along the deterministic alive branches (in the sense of the above $\lambda_{u}^{v}$ for deterministic machines). These substitutions correspond to compositions of copyless substitutions, hence they are copyless;
\item the general shape of the alive forest, i.e. how the deterministic alive branches are connected to each other, and what are the accessible states $Q_u$. This is a finite information that can be stored in a state.
\end{itemize}

\subparagraph*{Storing copyless substitutions.} We now explain how $\srans$ will "keep in memory" copyless substitutions, while being itself copyless. Let $s \in \subst{\regs}{B}$ be copyless, it can be described as:
\begin{itemize}
\item a function $\ske_s \colon \regs \rightarrow \regs^*$ describing where each
register is used in $s$.  Formally, $\ske_s \colon x \mapsto \mu(s(x))$ where
$\mu\colon (\regs \cup B)^* \rightarrow \regs^*$ is the morphism erasing the
letters.  Note that there is only a finite number of possible skeletons for
copyless substitutions;

\item a function $\beg_s \colon \regs \rightarrow B^*$ which maps $x$ to the 
word added "at the beginning" of the new $x$. Formally, it corresponds to the longest prefix of $s(x)$ that is in $B^*$;

\item a function $\fol_s \colon \regs \rightarrow B^*$ which maps $x$ to the
word added "after" the old $x$.  Indeed, given $x \in \regs$, there exists at
most one $y$ such that $x$ occurs in $s(y)$.  Then we define $\fol_s(x)$ as the
longest word in $B^*$ such that $s(y)= u x \fol_s(y) v$ for $u,v \in
(B \cup \regs)^*$.
\end{itemize}

\begin{example} \label{ex:subi} If $\regs = \{x,y\}$, let $s_1 = x \mapsto a, y \mapsto bxyc$ and $s_2:= x \mapsto yd, y \mapsto x$. We have:
\begin{itemize}
\item  $\ske_{s_1} = x \mapsto \vide, y \mapsto xy$ and $\beg_{s_1} = x \mapsto a, y \mapsto b$ and $\fol_{s_1} = x \mapsto \vide, y \mapsto c $;
\item $\ske_{s_2} = x \mapsto y, y \mapsto x$ and $\beg_{s_2} = x \mapsto \vide, y \mapsto \vide$ and $\fol_{s_2} = x \mapsto \vide, y \mapsto d $.
\end{itemize}
\end{example}

In other words, a copyless substitution can always be described by a bounded information ($\ske_s$) plus a finite number of words ($\beg_s(x)$ and $\fol_s(x)$ for $x \in \regs$). We shall thus keep $\ske_s$ in the finite memory (states) of $\srans$ and store the $\beg_s(x)$ and $\fol_s(x)$ for $x \in \regs$ in several registers of $\srans$.

With this representation, $\srans$ is able to perform a "virtual composition" of substitutions, by generalizing the example below. Furthermore, we claim that these virtual compositions can be performed \emph{in a copyless way} (with respect to the registers of $\srans$).

\begin{example} \label{ex:comp:ske} Following Example \ref{ex:subi}, we have $s_1 \circ s_2 = x \mapsto bxycd, y \mapsto a$ and:
\begin{itemize}
\item $\beg_{s_1 \circ s_2}(x) = \beg_{s_2}(x)  \beg_{s_1}(y) = \vide b = b$;
\item $\beg_{s_1 \circ s_2}(y) = \beg_{s_2}(y) \beg_{s_1}(x) \fol_{s_2}(x)= \vide a \vide = a$;
\item $\fol_{s_1 \circ s_2}(x) =\fol_{s_1}(x) = \vide$;
\item $\fol_{s_1 \circ s_2}(y) = \fol_{s_1}(y) \fol_{s_2}(y)= cd $;
\item $\ske_{s_1 \circ s_2}(x) = xy$ and $\ske_{s_1 \circ s_2}(y) = \vide$.
\end{itemize}

\end{example}

\subparagraph*{Updates of $\srans$.} 
Now that the representation of our abstraction by $\srans$ is understood, we 
explain how it is updated when reading a letter.  We shall follow the example given
by Figure \ref{fig:composition} to detail the evolution of the alive forest (the
notations are the same as in Figure \ref{fig:branches}).

Subfigure \ref{sub:a} presents the forest of initial runs after reading some
input $u \in A^*$.  When reading a new letter $a$, $\srans$ first computes the
successors of $Q_u$ in $\utrans$, and the substitutions applied along these
transitions.  These informations are depicted in Subfigure \ref{sub:b}: the upper
state had one successor but not the lower one.  Then $\srans$ notes that the
lower subtree is now dead: it can discard all the informations that concern it,
as shown in Subfigure \ref{sub:c}.  Now the alive forest only has $1$
deterministic branch, but $\srans$ still has $3$ substitutions.  In Subfigure
\ref{sub:d}, $\srans$ has composed these substitutions in order to keep a single
one in memory.  This "virtual" composition is implemented by updating the
registers that contain the $\beg$ and $\fol$ information, and the $\ske$, as
shown in Example \ref{ex:comp:ske}.

\begin{figure}[h!]
\begin{center}

	\begin{subfigure}[b]{0.4\textwidth}
	
	\begin{tikzpicture}[scale=1]

        \fill[red!20] (3.2,-0.6) -- (3.2,-0.4) -- (3.8,-0.4) -- (3.8,-0.6) -- cycle;
        		\fill[red!20] (4.2,-0.6) -- (4.2,-0.4) -- (5.8,-0.4) -- (5.8,-0.6) -- cycle;
		\fill[red!20]  (4.2,-0.65) -- (4.2,-0.85) -- (4.8,-1.1) -- (5.8,-1.1) -- (5.8,-0.9) -- (4.8,-0.9) -- cycle;
    	
	\node (in) at (3,-0.5) []{ $\bullet$};
	\draw[->,thick](3.2,-0.5) to (3.8,-0.5);		
	\node (in) at (4,-0.5) []{ $\bullet$};
	\draw[->,thick](4.2,-0.5) to (4.8,-0.5);			
	\node (in) at (5,-0.5) []{ $\bullet$};
	\draw[->,thick](5.2,-0.5) to (5.8,-0.5);	
	\node (in) at (6,-0.5) []{ $\bullet$};

		\draw[->,thick](4.2,-0.75) to (4.8,-1);		
		\node (in) at (5,-1) []{  $\bullet$};
		\draw[->,thick](5.2,-1) to (5.8,-1);		
		\node (in) at (6,-1) []{  $\bullet$};

       \end{tikzpicture}
			
	\caption{\label{sub:a} Forest after reading $u$}   
	
	\end{subfigure}
	~~~
	\begin{subfigure}[b]{0.4\textwidth}
	
	\begin{tikzpicture}[scale=1]

        \fill[red!20] (3.2,-0.6) -- (3.2,-0.4) -- (3.8,-0.4) -- (3.8,-0.6) -- cycle;
        		\fill[red!20] (4.2,-0.6) -- (4.2,-0.4) -- (5.8,-0.4) -- (5.8,-0.6) -- cycle;
		\fill[red!20] (6.2,-0.6) -- (6.2,-0.4) -- (6.8,-0.4) -- (6.8,-0.6) -- cycle;
		\fill[red!20]  (4.2,-0.65) -- (4.2,-0.85) -- (4.8,-1.1) -- (5.8,-1.1) -- (5.8,-0.9) -- (4.8,-0.9) -- cycle;
    	
	\node (in) at (3,-0.5) []{ $\bullet$};
	\draw[->,thick](3.2,-0.5) to (3.8,-0.5);		
	\node (in) at (4,-0.5) []{ $\bullet$};
	\draw[->,thick](4.2,-0.5) to (4.8,-0.5);			
	\node (in) at (5,-0.5) []{ $\bullet$};
	\draw[->,thick](5.2,-0.5) to (5.8,-0.5);	
	\node (in) at (6,-0.5) []{ $\bullet$};
	\draw[->,thick](6.2,-0.5) to (6.8,-0.5);	
	\node (in) at (7,-0.5) []{ $\bullet$};

		\draw[->,thick](4.2,-0.75) to (4.8,-1);		
		\node (in) at (5,-1) []{  $\bullet$};
		\draw[->,thick](5.2,-1) to (5.8,-1);		
		\node (in) at (6,-1) []{  $\bullet$};

       \end{tikzpicture}
			
	\caption{\label{sub:b} The successors of $Q_u$ are computed}   
	
	\end{subfigure}
	
	\vspace*{0.6cm}

	\begin{subfigure}[b]{0.4\textwidth}
	
	\begin{tikzpicture}[scale=1]

        \fill[red!20] (3.2,-0.6) -- (3.2,-0.4) -- (3.8,-0.4) -- (3.8,-0.6) -- cycle;
        		\fill[red!20] (4.2,-0.6) -- (4.2,-0.4) -- (5.8,-0.4) -- (5.8,-0.6) -- cycle;
		\fill[red!20] (6.2,-0.6) -- (6.2,-0.4) -- (6.8,-0.4) -- (6.8,-0.6) -- cycle;

	\node (in) at (3,-0.5) []{ $\bullet$};
	\draw[->,thick](3.2,-0.5) to (3.8,-0.5);		
	\node (in) at (4,-0.5) []{ $\bullet$};
	\draw[->,thick](4.2,-0.5) to (4.8,-0.5);			
	\node (in) at (5,-0.5) []{ $\bullet$};
	\draw[->,thick](5.2,-0.5) to (5.8,-0.5);	
	\node (in) at (6,-0.5) []{ $\bullet$};
	\draw[->,thick](6.2,-0.5) to (6.8,-0.5);	
	\node (in) at (7,-0.5) []{ $\bullet$};

		\draw[->,thick,gray](4.2,-0.75) to (4.8,-1);		
		\node[gray] (in) at (5,-1) []{  $\bullet$};
		\draw[->,thick,gray](5.2,-1) to (5.8,-1);		
		\node[gray] (in) at (6,-1) []{  $\bullet$};

       \end{tikzpicture}
			
	\caption{\label{sub:c} Dead subtrees are discarded}   
	
	\end{subfigure}
	~~~
	\begin{subfigure}[b]{0.4\textwidth}
	
	\begin{tikzpicture}[scale=1]

        \fill[red!20] (3.2,-0.6) -- (3.2,-0.4) -- (6.8,-0.4) -- (6.8,-0.6) -- cycle;
       
	\node (in) at (3,-0.5) []{ $\bullet$};
	\draw[->,thick](3.2,-0.5) to (3.8,-0.5);		
	\node (in) at (4,-0.5) []{ $\bullet$};
	\draw[->,thick](4.2,-0.5) to (4.8,-0.5);			
	\node (in) at (5,-0.5) []{ $\bullet$};
	\draw[->,thick](5.2,-0.5) to (5.8,-0.5);	
	\node (in) at (6,-0.5) []{ $\bullet$};
	\draw[->,thick](6.2,-0.5) to (6.8,-0.5);	
	\node (in) at (7,-0.5) []{ $\bullet$};

		\draw[->,thick,gray](4.2,-0.75) to (4.8,-1);		
		\node[gray] (in) at (5,-1) []{  $\bullet$};
		\draw[->,thick,gray](5.2,-1) to (5.8,-1);		
		\node[gray] (in) at (6,-1) []{  $\bullet$};

       \end{tikzpicture}
			
	\caption{\label{sub:d} Substitutions along alive deterministic branches are composed}   
	
	\end{subfigure}

\end{center}
       \caption{\label{fig:composition} Computing the alive deterministic branches after reading $a$ of input $ua$.}
\end{figure}
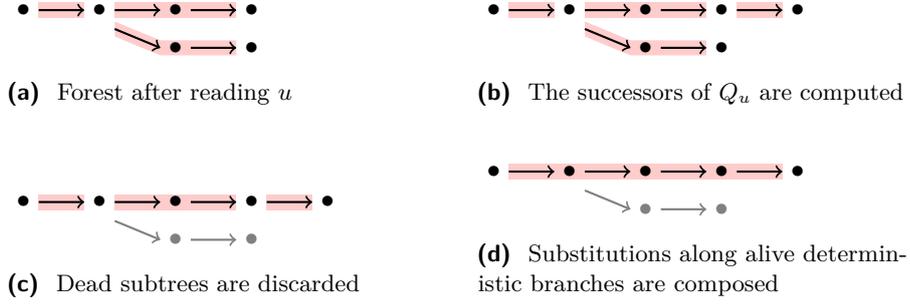

\subparagraph*{Output of $\srans$.}  
Once $\srans$ has read its whole input $w \in A^*$, there is exactly one state
$q \in Q_{w}$ which is final, due to unambiguity.  In its output function,
$\srans$ computes the contents of the registers after the single accepting run
of $\utrans$ (this is done by composing the substitutions of the deterministic
branches) and follows the output function of $\utrans$.

\subparagraph*{Use of external functions.} 
We did not discuss how external functions are used in the previous
constructions.  In fact, we treat a function name $\exte \in \oras$ as any
letter letter $b \in B$.  The machine $\srans$ calls it when it is used in the
new substitution that appears in the end of an alive branch.  The calls are done
at the same position in $\utrans$ and in $\srans$ (in some sense, we preserve
the "origin semantics" of the calls), hence they return the same values.

This idea
is detailed in the example below.

\begin{example} Assume that $\regs = \{x\}$ and $\oras = \{\exte, \exteg\}$. After reading $u \in A^*$, we suppose that only one deterministic branch is alive, which goes from an initial state to $\{q\} = Q_u$. The substitution applied along this branch is $s = x \mapsto \beg_s(x)x\fol_s(x)$. The machine $\srans$ keeps it in two registers $\beg$ and $\fol$ (together with the information $\ske_s$).

Let $a \in A$ be the next letter of the input, and suppose that $(q,a,q') \in
\Delta$ is the only outgoing transition from $q$ labelled by $a$.  Then after
reading $a$ we have $Q_{ua} = \{q'\}$, and the alive forest still contains only
one deterministic branch.  The substitution applied along this branch (that
$\srans$ has to compute) is now $s' := s \circ \lambda_{ua}$.

Assume that $\Lambda(q,a,q') = x \mapsto \exte x b\exteg$ with $b \in B$. Then $s' =  x \mapsto \exte(ua) \beg_{s}(x) x \fol_s(x) b \exteg(ua)$. Therefore $\srans$ performs the copyless updates $\beg \mapsto \exte \beg$ and $\fol \mapsto \fol  b \exteg$. The updates would be the same if $\exte$ and $\exteg$ were letters from $B$ instead of function names from $\oras$.
\end{example}

\end{document}